\documentclass[manuscript,acmsmall,screen]{acmart}
\AtBeginDocument{%
  }

\setcopyright{acmlicensed}
\copyrightyear{2018}
\acmYear{2018}
\acmDOI{XXXXXXX.XXXXXXX}
\acmConference[Conference acronym 'XX]{Make sure to enter the correct
  conference title from your rights confirmation email}{June 03--05,
  2018}{Woodstock, NY}
\acmISBN{978-1-4503-XXXX-X/2018/06}

\usepackage{tikz} 
\usetikzlibrary{arrows,shapes,automata,positioning,decorations,calc} 
\usetikzlibrary{spy,backgrounds,fit,shapes.misc}
\pgfdeclaredecoration{complete sines}{initial}
{
  \state{initial}[
  width=+0pt,
  next state=sine,
  persistent precomputation={\pgfmathsetmacro\matchinglength{
      \pgfdecoratedinputsegmentlength / int(\pgfdecoratedinputsegmentlength/\pgfdecorationsegmentlength)}
    \setlength{\pgfdecorationsegmentlength}{\matchinglength pt}
  }] {}
  \state{sine}[width=\pgfdecorationsegmentlength]{
    \pgfpathsine{\pgfpoint{0.25\pgfdecorationsegmentlength}{0.5\pgfdecorationsegmentamplitude}}
    \pgfpathcosine{\pgfpoint{0.25\pgfdecorationsegmentlength}{-0.5\pgfdecorationsegmentamplitude}}
    \pgfpathsine{\pgfpoint{0.25\pgfdecorationsegmentlength}{-0.5\pgfdecorationsegmentamplitude}}
    \pgfpathcosine{\pgfpoint{0.25\pgfdecorationsegmentlength}{0.5\pgfdecorationsegmentamplitude}}
  }
  \state{final}{}
}
\usepackage{bm}
\usepackage{booktabs}
\usepackage{amsmath,amsfonts}
\usepackage{algorithm}
\usepackage{algpseudocode}
\usepackage{graphicx}
\usepackage{textcomp}
\usepackage{pgfplots}
\usepackage{xcolor}
\usepackage{subcaption}

\usepackage{booktabs}

\usepackage{algorithm}

\usepackage{tikz}
\usetikzlibrary{calc}
\usetikzlibrary{matrix, positioning}
\usetikzlibrary{pgfplots.groupplots}
\usetikzlibrary{arrows}
\usepackage{tzplot}

\usepackage{acronym}
\usepackage[font=small,skip=0pt]{caption}

\newcommand{\ignore}[1]{{}}

\newcommand{\SF}{Simulink/Stateflow}

\acrodef{CPS}{Cyber-physical Systems}
\acrodef{STL}{Signal Temporal Logic}
\acrodef{TFL}{Time Frequency Logic}
\acrodef{TA}{Timed Automata}
\acrodef{FSM}{Finite State Machine}
\acrodefplural{FSM}{Finite State Machine}
\acrodef{DES}{Discrete Event Simulation}
\acrodef{DSP}{Digital Signal Processor}
\acrodef{DTTS}{Discrete Time Transition System}
\acrodef{EA}{Evolutionary Algorithm}
\acrodef{HA}{Hybrid Automata}
\acrodef{HIOA}{Hybrid Input Output Automaton}
\acrodefplural{HIOA}{Hybrid Input Output Automata}
\acrodef{ILP}{Integer Linear Programming}
\acrodef{MCU}{Microcontroller Unit}
\acrodef{ODE}{Ordinary Differential Equation}
\acrodef{PoC}{Plant-on-a-Chip}
\acrodef{PLL}{Phase-locked loop}
\acrodef{QSS}{Quantized State System}
\acrodef{QSHA}{Quantized State Hybrid Automata}
\acrodef{SA}{Sinoatrial}
\acrodef{SHIOA}{Synchronous Hybrid Input Output  Automata}
\acrodef{SWIOA}{Synchronous Witness Input Output Automata}
\acrodef{TTS}{Time Triggered System}
\acrodef{WHIOA}{Well-formed Hybrid Input Output Automata}
\acrodef{WCET}{Worst-Case Execution Time}
\acrodef{WCRT}{Worst-Case Reaction Time}
\acrodef{SEA}{Synchronous Emulation Automaton}
\acrodefplural{SEA}{Synchronous Emulation Automata}
\acrodef{SWA}{Synchronous Witness Automaton}
\acrodefplural{SWA}{Synchronous Witness Automata}
\acrodef{DSL}{Domain Specific Language}
\acrodef{LTI}{Linear Time Invariant}
\acrodef{SMT}{Satisfiability Modulo Theory}
\acrodef{BMC}{Bounded Model Checking}
\acrodef{DHIOA}{Discrete Hybrid Input Output  Automata}
\acrodef{RHS}{Right Hand Side}
\acrodef{LHS}{Left Hand Side}
\acrodef{TH}{Thermostat}
\acrodef{WLM}{Water Level Monitor}
\acrodef{AFb}{Atrial Fibrillation}
\acrodef{Reactor}{Nuclear Reactor}
\acrodef{Robot}{Nonholonomic Robot}
\acrodef{RK}{Runge-Kutta}
\acrodef{DASSL}{Differential Algebraic System Solver}
\acrodef{DQSS}{Dynamic Quantized State System}
\acrodef{QSHIOA}{Quantized State Hybrid Input Output Automaton}
\acrodef{SWC}{Steering Wheel Control}
\acrodef{FA}{Frequency Automata}

\newtheorem{definition}{Definition}
\newtheorem{theorem}{Theorem}[section]
\newtheorem{lemma}[theorem]{Lemma}
\newcounter{remark}

\algrenewcommand\algorithmicrequire{\textbf{Input:}}
\algrenewcommand\algorithmicensure{\textbf{Output:}}
\usepackage{stmaryrd}
\newcommand\numberthis{\addtocounter{equation}{1}\tag{\theequation}}

\usepackage{wrapfig}
\algdef{SE}[DOWHILE]{Do}{doWhile}{\algorithmicdo}[1]{\algorithmicwhile\ #1}%

\usepackage{enumitem}

\begin{document}

\title{Frequency Automata: A novel formal model of hybrid systems in combined time and frequency domains}

\author{Moon Kim}
\email{moon.kim@auckland.ac.nz}
\author{Avinash Malik}
\email{avinash.malik@auckland.ac.nz}

\author{Partha Roop}
\authornote{Partha is the corresponding author}
\email{p.roop@auckland.ac.nz}

\affiliation{%
  \institution{The University of Auckland}
  \city{Auckland}
  \country{New Zealand}
}

\renewcommand{\shortauthors}{Kim et al.}

\begin{abstract}
  Hybrid systems are mostly modelled, simulated, and verified in the
  time domain by computer scientists. Engineers, however, use both
  frequency and time domain modelling due to their distinct advantages.
  For example, frequency domain modelling is better suited for control
  systems, using features such as spectra of the signal. Considering
  this, we introduce, for the first time, a formal model called
  frequency automata for hybrid systems modelling and simulation, which
  are represented in combined time and frequency domains. We propose a
  sound translation from \ac{HA} to \acf{FA}. We also develop a
  numerical simulator for \ac{FA} and compare it with the performance of
  \ac{HA}. Our approach provides precise level crossing detection and
  efficient simulation of hybrid systems. We provide empirical results
  comparing simulation of \ac{HA} via its translation to \ac{FA} and its
  simulation via Matlab Simulink/Stateflow\textregistered. The results show clear
  superiority of the proposed technique with the execution times of the
  proposed technique $118\times$ to $1129\times$ faster compared to
  Simulink/Stateflow\textregistered. Moreover, we also observe that the proposed
  technique is able to detect level crossing with complex guards
  (including equality), which Simulink/Stateflow\textregistered\ fail.
\end{abstract}

\ccsdesc[500]{Do Not Use This Code~Generate the Correct Terms for Your Paper}
\ccsdesc[300]{Do Not Use This Code~Generate the Correct Terms for Your Paper}
\ccsdesc{Do Not Use This Code~Generate the Correct Terms for Your Paper}
\ccsdesc[100]{Do Not Use This Code~Generate the Correct Terms for Your Paper}

\keywords{Do, Not, Us, This, Code, Put, the, Correct, Terms, for,
  Your, Paper}

\received{20 February 2007}
\received[revised]{12 March 2009}
\received[accepted]{5 June 2009}

\maketitle

\section{Introduction}
\label{sec:introduction}

Formal methods~\cite{woodcock2009formal} encompass a set of
mathematically founded techniques for specification, verification and
synthesis. These are especially suited for the design of
\acfp{CPS}~\cite{alur2015principles}.\ \acp{CPS} utilise the control of
physical processes, using distributed embedded controllers. Here,
analogue values are received from sensors, which are processed using
different signal processing techniques, before the clean signal is
utilised as an input by the controller. As \acp{CPS} are highly
safety-critical, there is a need for formal modelling and analysis of
such signals.

The study of signals, signal representation and analysis forms the core
of a discipline of electrical engineering called \emph{signal
  processing}~\cite{devasahayam2019signals}. In the time domain, the
changes in the signal's value can be observed over time. In contrast, in
the frequency domain features such as energy, power, and spectra can be
considered over different frequency ranges. Engineers use both time
domain and frequency domain methods for signal processing as both
provide distinct advantages. Frequency domain analysis, is especially
useful to determine the phase changes of a signal, for removal of
certain undesirable frequency components through filtering, and for
control purposes, such as to determine the controllability and stability
of a system.

In contrast, computer scientists are mostly interested in the time
domain representation. For example, majority of formal models are either
un-timed, logically timed or physically timed, represented in the time
domain~\cite{baier2008principles}. Un-timed models include automata and
logics such as finite automata and propositional
logic~\cite{baier2008principles}, for example. Logically timed models,
also known as synchronous models~\cite{benveniste03}, form the basis of
synchronous programming languages. In contrast, physical time can be
captured in formal models such as timed automata~\cite{alur1994theory}
and hybrid automata~\cite{alur2015principles}. Majority of mathematical
logics since the pioneering work of Pnueli~\cite{pnueli1977temporal},
who provided the relationship between programs and temporal logics, have
been classified as either linear-time or branching time logics, in the
time domain~\cite{baier2008principles}.

With the advent of \ac{CPS}, there has been considerable interest in
formal representation and verification of signals, using
\acf{STL}~\cite{maler2004monitoring}, in 2004 by the pioneering work of
Maler and Nickovic. \ac{STL} provides a time domain representation of a
signal, $x(t)$ so that we can make assertions about its amplitude and
periodicity, etc. Since 2004, there is a large body of work on \ac{STL}
and its applications in \ac{CPS}~\cite{donze2013signal, su2025runtime}.
However, only recently there has been interest in formalising the
frequency domain representation. To the best of our knowledge, the first
work in this direction is an unified framework for both time and
frequency domain using \acf{TFL}~\cite{donze2012temporal}. Here, a
signal is represented using both a time window and a frequency domain
representation using short time frequency transform (STFL). In the logic
\ac{TFL}, \ac{STL} is extended with a spectral signal
$y = f_{L,\omega}(x)$, where $f_{L,\omega}(x)$ is an operator that produces the
projection of the $L-$spectrogram of signal $x$ on frequency
$\omega$~\cite{donze2012temporal}. More recently, \ac{TFL} has been used in
applications such as detection of signal
abnormalities~\cite{nguyen2017abnormal}. Two other recent approaches to
seamless analysis in both time and frequency domain are the
consideration of the filtering operation from a logical
perspective~\cite{rodionova2016temporal} and a logical approach to
signal processing~\cite{basnet2020logical}, which uses Volterra series
for monitoring temporal logic formulae in the frequency domain.

\subsection{Gap and hypothesis}

All the above formulations to leverage time and frequency domain
analysis are in the logical realm, and are mostly \ac{STL} extensions.
We would like to assert that similar formulations are needed in the
automata realm, especially to consider the representation of hybrid
systems, using well known formalisms such as the hybrid automata. To the
best of our knowledge, there is no work that provides a frequency domain
representation of \acf{HA}. A question may arise, why this may be
needed? We hypothesise that a frequency domain mapping of \ac{HA} will
enable more precise simulation and analysis of hybrid systems compared
to the time domain model.

\subsection{Intuition and overview of the solution}

A hybrid automaton combines continuous evolution of a set of real-valued
variables, which represent signals, with discrete mode switches. A given
mode switch happens when a \emph{level crossing} is detected. A level
crossing is usually reached when a Boolean condition over the value of a
signal or a set of signals is satisfied.

\ac{HA}, however, suffers from many limitations. First the level
crossing detection problem is very challenging. When solved dynamically,
there are many challenges due to the need for back tracking. While
solving statically, level crossing is also missed, as the signal can
change values continuously within a time step. Second, the efficiency
and precision of simulation depends on the step size used, which has a
large bearing on the tools. By going to the frequency domain, we
hypothesise a solutions that overcomes these challenges.

We are inspired by an approach for modelling the action potential of cardiac cells, using a 
model called \emph{the resonant model}~\cite{sehgal2019resonant}, which uses a small number of terms in 
a Fourier series to approximate the function, that is combined with a discrete 
controller. However, the developed model is not formalised. 
Our intuition is the following. The evolution of signals is captured as
\acfp{ODE} in a \ac{HA}. A signal can also be represented by capturing
the same flow as a Sinusoid by going to the frequency domain as follows.
First we can map the flow constraint in an ODE to an equivalent
frequency representation using angular velocity (reciprocal of the
frequency) of the same signal, when projected in an unit circle. An
infinitesimal change in time can be mapped to an equivalent
infinitesimal change in phase angle. Also, instead of modelling the rate
of change as an ODE, we capture this as an equivalent change in angular
velocity. Finally, instead of performing numerical integration to
recover the signal, we take angular steps that converge exactly to the
level crossing.

As we move away from time based integration into the realm of phase and
time-based integration, level crossings are never missed. This is since,
as the vector moves in the unit circle by taking small step changes in
the phase angle, the level crossing will definitely be crossed before
the $2\pi$ revolution is completed. Moreover, we show that this approach
enables a dynamic step size solution, which not only precisely converges
to the guard, but also takes significantly fewer steps in the solution
simulation, compared to traditional \ac{HA} simulation.

We make the following assumptions. First, we limit ourselves to
\acp{ODE} of the form $\dot{x}(t) = f(x(t))$. Second, we consider a
subset of \ac{HA} which only exhibit switching behaviour i.e., the
intersection of any guard with the corresponding invariant is empty. The
main contributions of our approach are as follows:

\begin{enumerate}
\item We propose \acf{FA} as a novel formal model of hybrid systems in
  the frequency domain.
\item We propose a sound translation of \ac{HA} to \ac{FA}.
\item We show that we can precisely detect level crossings of hybrid
  systems, for the first time.
\item We propose a new method of numerical integration by utilising the
  \ac{FA} model, which is more efficient and precise compared to the
  existing methods of \ac{HA} simulation.
\end{enumerate}

The rest of the paper is arranged as follows. We start with a motivating
example in Section~\ref{sec:runn-example-probl}.
Section~\ref{sec:background} gives the background in infinitesimals and
non-standard analysis required to prove our propositions.
Section~\ref{sec:ha-syntax-semantics} gives the syntax and semantics of
the \ac{HA} and the \ac{FA}. In this section we also prove the
equivalence relation of the two semantics.
Section~\ref{sec:an-effic-simul} gives an efficient simulation algorithm
for the \ac{FA}. The benchmark results are provided in
Section~\ref{sec:experimental-results}. Related work describing the
current state-of-the-art in \ac{HA} simulation is presented in
Section~\ref{sec:related-work}. The conclusions and future work are
provided in Section~\ref{sec:concl-future-work}.

\section{Running example, problem description and the proposed solution}
\label{sec:runn-example-probl}

Hybrid systems comprise of a physical process, also known as the
\emph{plant}, which is controlled by an appropriate \emph{controller}.
While the plant exhibits continuous dynamics, the controller is executed
on an embedded platform and hence exhibits discrete behaviour. Matlab
Simulink/Stateflow\textregistered provides an industry grade software for the
modelling, simulation and code generation for such systems.
Figure~\ref{fig:steeringwheel} depicts a typical steering wheel of an
autonomous vehicle and its associated controller, which is our running
example.

In this example, the steering wheel angle is initially at
$\frac{\pi}{2}$ (Figure~\ref{fig:1a}). The angular velocity is positive
when the wheel rotates anti-clockwise and negative when it rotates clock
wise. In Figure~\ref{fig:1a}, the up arrow is the point that indicates
the initial angle of the steering wheel. The control objective is to
maintain the steering wheel rotation within limited range as indicated
with red markers. Figure~\ref{fig:1b} shows the overall
Simulink/Stateflow\textregistered\ implementation of the control system, with the
steering wheel plant model (Figure~\ref{fig:1c}) being controlled by a
discrete integrated controller. The model sets the initial values of the
steering wheel position $x(t)$ and its cosine ($y(t)$), respectively.
Then the Stateflow model starts from initial location \texttt{L1}. In
this location, the steering wheel position is continuously updated
anti-clockwise with an angular velocity of 0.1 radians/second. Once the
steering wheel position reaches the left red marker, indicated by
$\cos(x(t)) \leq -0.99$, then an instantaneous switch is made to location
\texttt{L2} and the model starts evolving the steering wheel in the
clock-wise direction with angular velocity of -4 radians/second. This
system evolves forever trying to maintain the position of the steering
wheel between the two red markers.

\begin{figure}[tb]
  \subfloat[Steering wheel~\label{fig:1a}]{
    \scalebox{0.45}{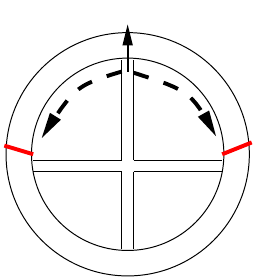}
  }
  \qquad
  \subfloat[The overall Stateflow implementation of the steering wheel
  control~\label{fig:1b}]{
    \includegraphics[scale=0.22]{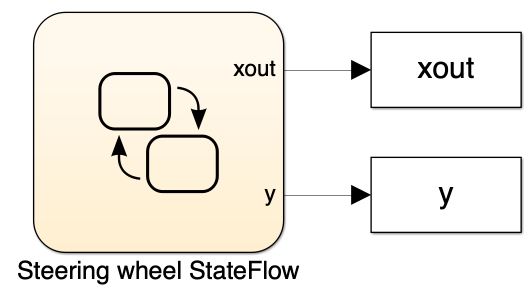}
  }
  \qquad
  \subfloat[The Stateflow model~\label{fig:1c}]{
    \includegraphics[scale=0.3]{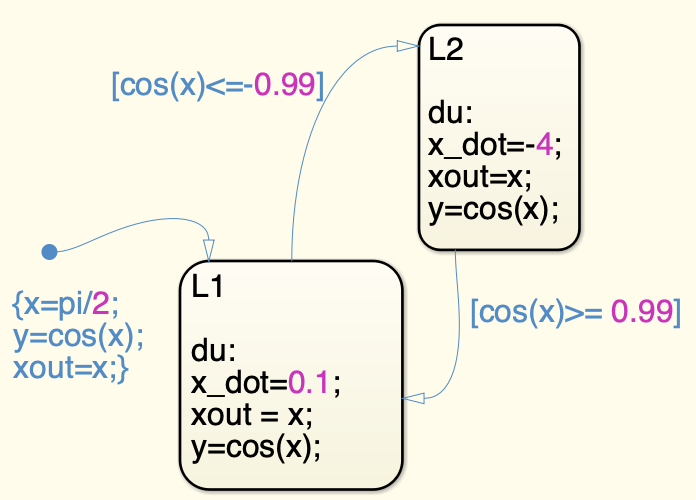}
  }
  \caption{Steering wheel of a car and its implementation in Stateflow}
  \label{fig:steeringwheel}
\end{figure}

\begin{figure}[tbh]
  \centering
  \subfloat[Incorrect wheel angle~\label{fig:2a}]{
    \includegraphics[scale=0.11]{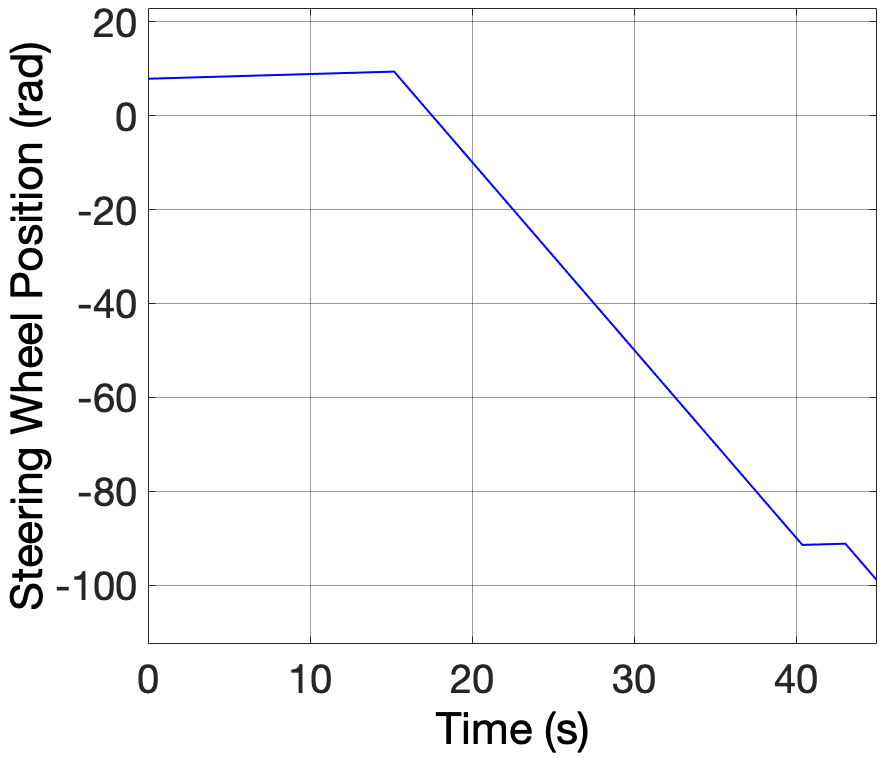}
  }
  \subfloat[Incorrect level crossings~\label{fig:2b}]{
    \includegraphics[scale=0.11]{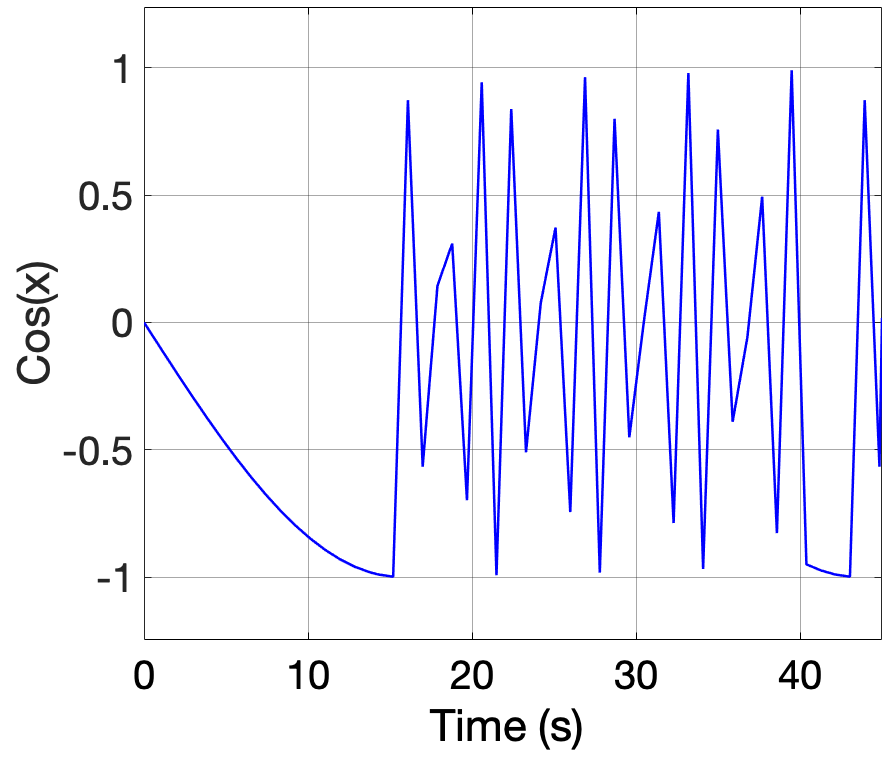}
  }

  \subfloat[Correct wheel angle~\label{fig:2c}]
  {
    \includegraphics[scale=0.21]{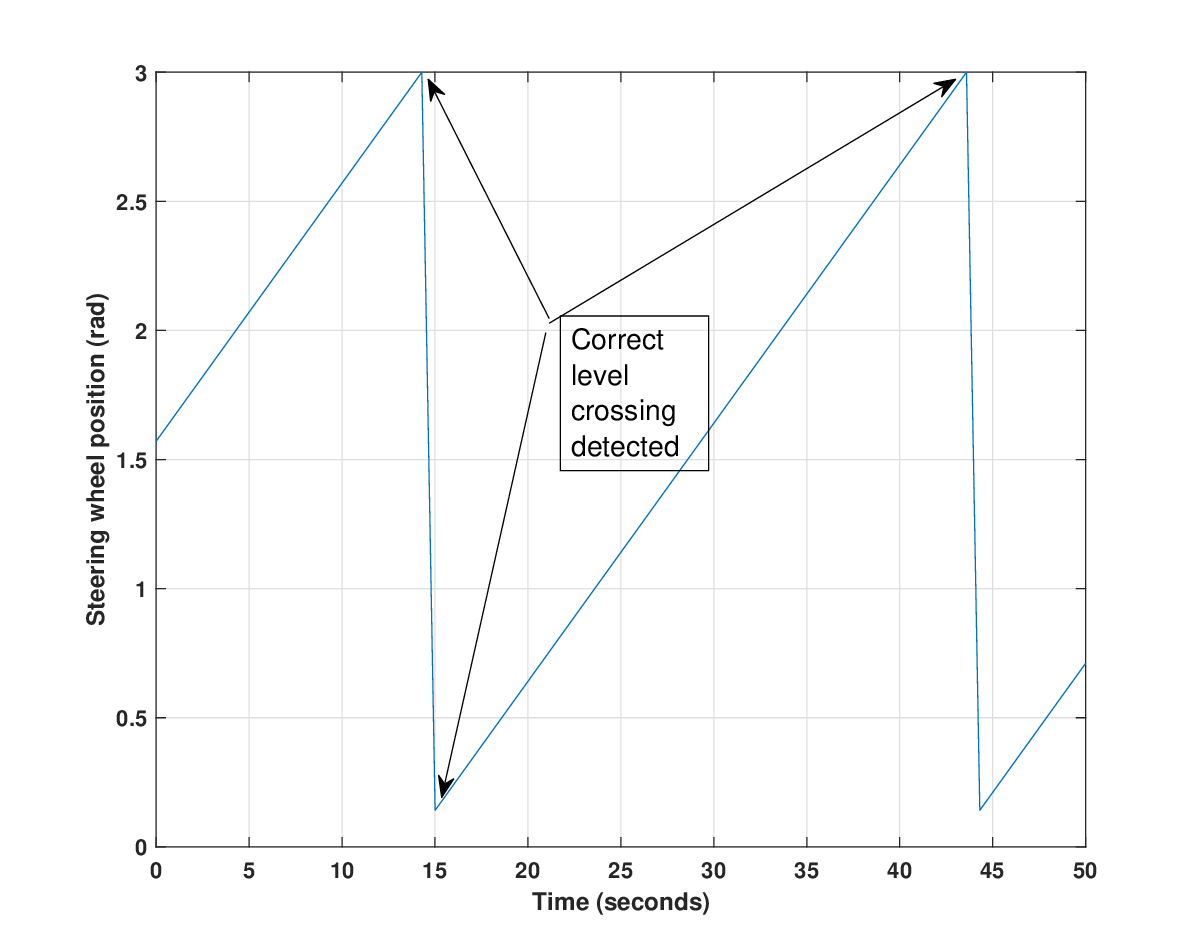}
  }
  \subfloat[Correct level crossings~\label{fig:2d}]
  {
    \includegraphics[scale=0.21]{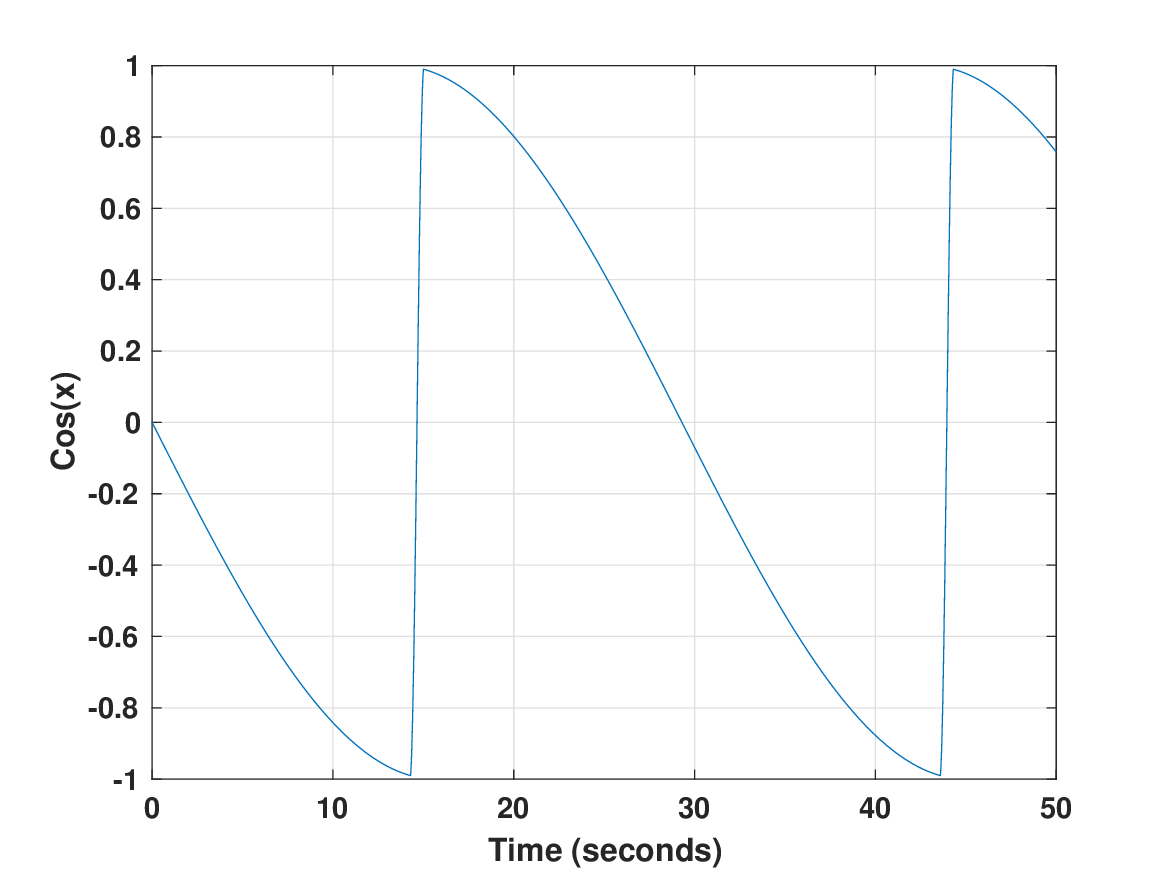}
  }
  \caption{Outputs from the Stateflow model for a $50\sec$ simulation
    time}
  \label{fig:stateflowoutput}
\end{figure}

Figure~\ref{fig:stateflowoutput} shows the different traces generated
from the Stateflow model by using different solvers. First, the default
setting uses a variable step-size ODE45 solver. We set the error
tolerance to $10^{-12}$. The output trace of the steering wheel angle
$x(t)$ is erroneous as shown in Figure~\ref{fig:2a}. It is expected that
$x(t)$ is bounded, however, the result shows that the angle keeps
increasing negatively. The reason for the erroneous trace becomes
obvious by looking at the output $\cos(x(t))$, the level crossing guard
condition in Figure~\ref{fig:2b}. The simulation engine is unable to
find the \textit{very} first level-crossing when the condition
$\cos(x(t)) \geq 0.99$ or $\cos(x(t)) \leq -0.99 $ holds.

This is the well known \textbf{even number of level-crossing detection}
problem~\cite{zhang2008zero}. The level-crossing conditions
$\cos(x) \leq -0.99 $ and $\cos(x) \geq 0.99$ hold every $2\times\pi$ radians. The
level-crossing detector instead of finding the very first level-crossing
ends up detecting some level-crossing at $(2n\pi)+ x$ radians, for some
natural $n > 1$.

Reducing the error tolerance further as suggested in the \SF\textregistered\
documentation still gives a wrong level-crossing value. The only working
solution is to manually reduce the step-size in order to find the very
first level-crossing\footnote{Changing the solver to a stiff system
  solver has no affect.}. Figures~\ref{fig:2c} and~\ref{fig:2d} show the
correct steering wheel angle and correct level-crossings being detected
with a maximum (fixed) step-size of $10^{-2}$. Step-sizes of greater
than $10^{-2}$ also give incorrect results. However, decreasing the
step-size leads to a significant slowdown in the execution speed of the
model. The variable step-size solver with the default settings takes 62
simulation steps, while with a fixed step-size of $10^{-2}$ ends up
taking 5000 simulation steps, i.e., an overall slowdown of around $80$
times.

The \ac{HA} representation of the steering wheel system is shown in
Figure~\ref{fig:stwheel}. Following the Stateflow model, there is a
single starting location L1 with initial values set before entering the
initial location (shown with the dashed arrow on location L1). The model
remains in this location while the invariant holds, and the outgoing
edge guards are not triggered. As soon as the outgoing edge guard is
triggered, the model makes an instantaneous jump to location L2, during
this jump the reset actions set the new values of the continuos
variables --- in this example the continuos variables carry their values
over. Upon entering location L2, the continuous variables start
evolving again according to their respective \acp{ODE} and updates.

\begin{figure}[thb]
  \centering
  \scalebox{0.9}{\begin{tikzpicture}[->,>=stealth',shorten >=1pt,auto,
node distance=6.4cm,
semithick,scale=0.7, transform shape]
\tikzstyle{every state}=[rectangle,rounded corners, minimum height =
1.2cm, text width=2.8cm, text centered,
fill=blue!20,draw=none,text=black, draw,line width=0.3mm]

\node[state,
label={[shift={(0,0.1)}]$\neg(y(t) \leq -0.99)$}, 
label={[shift={(1.5,0.1)}]\textbf{\texttt{L1}}}]
(T2)  {
  $\dot{x}(t) = 0.1$ \\
  $y(t) = cos(x(t))$
};

\node[state, 
label={[shift={(0,0.1)}]$ \neg (y(t) \geq 0.99)$}, 
label={[rotate=0,shift={(-1.5,0.1)}]\textbf{\texttt{L2}}  }] 
(T3) [node distance=6.2cm, right of=T2]
{
  $\dot{x}(t) = -4$ \\
  $y(t) = cos(x(t))$
};

\draw[transform canvas={yshift=0.3em}] (T2) -- (T3) node [midway]
{$\frac{y(t) \leq -0.99}{x'(t)=x(t), y'(t)=y(t)}$};

\draw[transform canvas={yshift=-0.3em}] (T3) -- (T2) node [midway]
{$\frac{y(t) \geq 0.99}{x'(t)=x(t),y'(t)=y(t)}$};

\draw[<-, dashed](T2.180) -- node[below] {} ++(-2cm,-0cm);
\draw[<-, dashed](T2.180) -- node[above] {
  $
  \begin{aligned}
    x(0) = \pi/2, y (0) = cos(x(0)), \\
  \end{aligned}
  $} ++(-4.5cm,-0cm);



\end{tikzpicture}

}
  \caption{\ac{HA} representation of the steering wheel system}
  \label{fig:stwheel}
\end{figure}
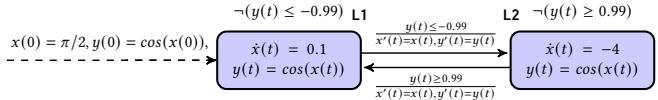

The proposed \ac{HA} execution technique described in
Section~\ref{sec:prop-simul-techn} simulates the steering wheel \ac{HA}
resulting in the correct outputs (Figures~\ref{fig:2c} and~\ref{fig:2d})
without missing the level crossings. This simulation takes 38 simulation
steps compared to the correct Simulink/Stateflow\textregistered\ simulation with 5000
steps, which is a speed up of $131 \times$.


\subsection{The proposed solution --- from time to frequency domain
  mapping}
\label{sec:prop-simul-techn}

\begin{figure}[tbh]
  \centering
  \includegraphics[scale=0.8]{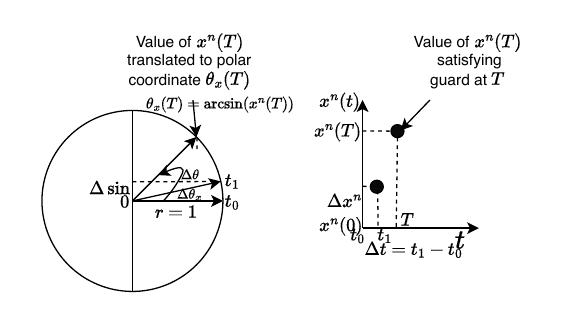}
  \caption{Translating cartresian coordinates to polar coordinates
    during translation of \ac{HA} to \ac{FA}}
  \label{fig:hatofa}
\end{figure}

\ignore{In this section we given an overview of our main ideas. The
formalisation and elaboration of our main ideas is presented in the
following sections.}

Consider the \acf{HA} model of the steering wheel (c.f.
Figure~\ref{fig:stwheel}), which consists of two locations $L1$ and
${L2}$. The continuos variable $x(t)$ describes the state of the
\ac{HA}. The variable $y(t)$ captures the cosine of $x(t)$ and is used
to capture the instantaneous switch between the two locations. The
variable $x(t)$ evolves using the \ac{ODE} (or equivalently vector flow)
$\dot{x}(t) = 0.1$ in location $L1$. This location is left as soon as
the outgoing guard $y(t) \leq -0.99$ is triggered on the edge connecting
${L1}$ to ${L2}$. Upon taking this guard the so called reset relations
($x'(t) = x(t)$ and $y'(t) = y(t)$) carry over the value of the two
variables over to location ${L2}$. The negation of the outgoing guard
acts as the invariant on location ${L1}$. We also have invariants,
\acp{ODE}, and outgoing guards and reset relations for location ${L2}$.

Before we present the main ideas of translating the \ac{HA} into
\acf{FA} we expound our assumptions. For a given \ac{HA} (e.g.,
Figure~\ref{fig:stwheel}) we assume that we know the following
\textit{statically} (without executing the \ac{HA}):

\begin{enumerate}
\item Initial value of every state variable when entering some location
  in the \ac{HA}. For example, we know that $x(0) = \pi/2$ when the
  steering wheel example enters location $L1$. Same with $L2$.
\item The (possibly symbolic) value of the state variable(s) that
  trigger one or more outgoing edge(s) from a given location. For the
  steering wheel example, we know that $x(T) = \arccos(-0.99)$ triggers
  the outgoing edge guard from location $L1$ at some time $T$.
\item The maximum value that any state variable can take in any given
  location. For the steering wheel example we have
  $\max(range(x(t))) = \max[\pi/2, abs(\arccos(-0.99))]$.
\end{enumerate}

Our translation from the time domain to the frequency domain proceeds as
follows:

\begin{enumerate}[label=Step-(\arabic*)]
\item Every location and edge in the \ac{HA} has an equivalent location
  and edge in the \ac{FA}.
\item For every state variable in the \ac{HA} we have a so called
  normalised state variable (from here on in called $x^{n}(t)$) in the
  \ac{FA}. Such that $x^{n}(t) = \frac{x(t) - x(0)}{\max(range(x(t)))}$
  for some state variable $x(t)$, in any given location. The initial and
  final value of the state variables (from assumptions above) are
  translated to 0 and some (possibly symbolic) value $x^{n}(T)$ that
  triggers the outgoing edge if any from the current location at time
  $T$ (c.f. Figure~\ref{fig:hatofa}).
\item Next, we consider a unit circle with centre 0 and radius
  ($r = 1$). We project the normalized variables as a vector rotating in
  the unit circle represented by the polar coordinates
  ($r\angle\theta(t)$). For the running example; the normalised variable
  $x^{n}(t)$ and its corresponding vector representation are shown in
  Figure~\ref{fig:hatofa}. At time $t_{0} = 0$, when entering location
  $L1$, we have $x^{n}(t) = 0$ and its corresponding vector is $\angle0$.

\item Next, comes our \textbf{main insight}: from the definition of
  derivatives
  $\lim_{\Delta t\rightarrow0}\frac{x^{n}(t+\Delta t) - x^{n}(t)}{\Delta t}$, we know that for an
  \textit{infinitesimal change} in time $\Delta t$,
  $\Delta x^{n} = x^{n}(t+\Delta t) - x(t)$ is constant. We capture this
  infinitesimal change as an infinitesimal change in the position of the
  vector on the unit circle. Hence,
  $\Delta x^{n} = \Delta\sin = \sin(\theta_{x} + \Delta\theta_{x}) - \sin(\theta_{x})$, where
  $\Delta\theta_{x} = \theta_{x}(t+ \Delta t) - \theta_{x}(t)$. The example of such a change in
  $x^{n}(t)$ from $t_{0}$ to $t_{1}$ being captured by the change
  $\Delta\sin$ is shown in Figure~\ref{fig:hatofa}, with the movement of the
  vector from its initial angle (0) at time $t_{0}$ to some angle
  $\Delta\theta_{x}$ at time $t_{1}$, such that $\Delta\sin = \Delta x^{n}$.
\item From the above insight we derive the function
  ($\mathcal{L}: \mathbb{R}\rightarrow\mathbb{R}$) capturing the \textit{dynamics} of
  $x^{n}(t)$ as the movement of the vector on the unit circle, i.e.,
  $\dot{\theta}(t) = \mathcal{L}(\dot{x}^{n}(t))$ (Lemma~\ref{lemma:eqstatevar}).
  Function $\mathcal{L}$ relates angular frequency ($\dot{\theta}(t)$) of the vector
  moving on the unit circle to the rate of change of the state variable
  $x^{n}(t)$. \textit{Unlike} Fourier transforms, this angular velocity
  is \textbf{not} constant. This non-constant angular velocity allows us
  to model the dynamics of $x^{n}(t)$ as the \underline{identity}
  $x^{n}(t) = \sin(\theta(t))$ (Lemma~\ref{lemma:eqstatevar}).
\item The above identity allows us to model the \textit{change} in the
  current value of $x^{n}(t)$ and the one satisfying the outgoing edge
  guard: $x^{n}(T) - x^{n}(t)$ as
  $\Delta\sin = \sin(\theta_x(T)) - \sin(\theta_{x}(t))$, which is in turn the change
  in the angle $\Delta\theta = \theta_{x}(T) - \theta_{x}(t)$. For the steering wheel
  example angle $\theta_{x}(T)$ that triggers the outgoing edge guard is
  shown in Figure~\ref{fig:hatofa} along with corresponding $x^{n}(T)$.
\item From non-standard analysis (Theorem~\ref{thm:inc}), angular change
  $\Delta\theta$ is a function ($\mathcal{G}$) of $\dot{\theta}(t)$, i.e.,
  $\Delta\theta = \mathcal{G}(\dot{\theta}(t))$. Hence, from point (5), we have
  $\Delta\theta = \mathcal{G}(\mathcal{L}(\dot{x}^{n}(t)))$.
\end{enumerate}

Steps (1)---(7) translate the dynamics of the continuous state-variable in
any given location of the \ac{HA} to the frequency of the rotating
vector on the unit circle --- the \ac{FA}. We do such a projection from
the time domain to the frequency domain, because we can capture the
level crossings, on out going edge guards, from any given location as
points on the unit circle (c.f. Figure~\ref{fig:hatofa}), which in turn
allows us to \textit{converge} to the level crossings exactly, in very
few simulation steps. This in turn, reduces the simulation time of the
\ac{HA}. The value of the continuous variable that triggers the outgoing
edge guard from any location is translated to points on the unit circle
(e.g., $\theta_{x}(T)$ in Figure~\ref{fig:hatofa} for the outgoing edge guard
from location ${L1}$ in the steering wheel example). Once this point is
determined; we can directly relate the change in the vector angle
$\Delta \theta$ to the level crossing, such that
$\Delta \theta = \theta_{x}(T) - \theta_{x}(t)$. Thereby guaranteeing convergence of the
vector to the level crossing (Lemma~\ref{lemma:guardconverge}). If the
vector converges to the level crossing, so does the state variable in
the \ac{HA}.


\section{Background}
\label{sec:background}

We utilise the machinery of \emph{non-standard analysis} developed by
Keisler~\cite{keisler2013elementary} to represent and manipulate
infinitesimals. Hence, in this section we provide the definitions and
theorems, of non-standard analysis, necessary to understand the rest of
the paper. For a complete treatment of non-standard analysis please
refer to~\cite{keisler2013elementary}.

\paragraph{\textbf{Non-standard analysis}}
\label{sec:non-stand-analys}


Non-standard analysis extends the real number line with the hyperreal
numbers. We start with giving the definition of the hyperreals.

\begin{definition}[Extension Principle]
\label{def:extensionp}
 First we define the \textit{hyperreal numbers}.
  \begin{itemize}
  \item The real numbers ($\mathbb{R}$) form the subset of \textit{hyperreal}
    numbers (${^{*}\mathbb{R}}$), and the order relation
    $x < y\ \forall x, y \in \mathbb{R}$, is a subset of the order relation for hyperreal
    numbers.
  \item There is a hyperreal number ($\epsilon$) that is greater than zero but
    less than every positive real number.
  \item For every real function $f$ of one or more variables we are
    given a corresponding hyperreal function ${^{*}f}$ of the same
    number of variables. Function ${^{*}f}$ is called the natural
    extension of $f$.
  \end{itemize}
\end{definition}


\begin{definition}
  A number $b \in {^{*}\mathbb{R}}$ is said to be:
  \begin{itemize}
  \item \textbf{positive infinitesimal} if $b$ is positive but less than
    every positive real number.
  \item \textbf{negative infinitesimal} if $b$ is negative but greater
    than every negative real number.
  \item \textbf{infinitesimal} if $b$ is either positive infinitesimal,
    negative infinitesimal, or zero.
  \end{itemize}
\end{definition}

\begin{definition}[Transfer Principle]
  \label{def:transfer}
  Every real statement that holds for one or more particular real
  function holds for the hyperreal natural extension of these functions.
\end{definition}

The transfer principle, allows us to state that every real function (or
a combination thereof) extends from the real numbers to the hyperreals
(c.f.\cite{keisler2013elementary}, Chapter-1).

\begin{definition}
  A hyperreal number $b \in {^{*}\mathbb{R}}$ is said to be:
  \begin{itemize}
  \item \textbf{finite} if $b$ is between two real numbers.
  \item \textbf{positive infinite} if $b$ is greater than every real
    number.
  \item \textbf{negative infinite} if $b$ is less than every real
    number.
  \end{itemize}
\end{definition}




\begin{definition}
  Two hyperreal numbers $b$ and $c$ are said to be \textbf{infinitely
    close} to each other (denoted $b \approx c$), if their difference $b-c$ is
  infinitesimal.
\end{definition}



\begin{definition}[Standard Part Principle]
  \label{def:st}
  Every finite hyperreal number is infinitely close to exactly one real
  number.
  Let $b$ be a finite hyperreal number. The \textbf{standard part} of
  $b$ (denoted $\mathrm{st}(b)$), is the real number which is infinitely
  close to $b$.
\end{definition}

We are now ready to give the definition of a derivative in terms of
hyperreal numbers.

\begin{definition}[Derivative]
  \label{def:deriv}
  Let $f(t)$ be a real function of one variable ($t$). The
  \textbf{derivative} of function $f(t)$ (denoted $\dot{f}(t)$) is a new
  function whose value at $t$ is given by
  \[
    \dot{f}(t) = \mathrm{st}\left (\frac{f(t + \Delta t)-f(t)}{\Delta t}\right )
  \] where $\Delta t$ is an infinitesimal and $\Delta t\ne 0$.
\end{definition}


\begin{definition}[Increment]
  Let $y = f(t)$, then $\Delta y = f(x + \Delta t) - f(t)$ is called the increment
  of $y$. Here $y$ is the dependent variable and $t$ is the independent
  variable.
  \label{def:inc}
\end{definition}

Given the increment, we can write the derivative in the short hand
$\dot{y} = \dot{f}(t) = \mathrm{st}(\frac{\Delta y}{\Delta t})$.

\begin{theorem}[Increment Theorem]
  \label{thm:inc}
  Let $y = f(t)$. Suppose $\dot{f}(t)$ exists at a certain $t$, and
  $\Delta t$ is infinitesimal with $\Delta t \neq 0$. Then $\Delta y$ is infinitesimal
  and; $\Delta y = \dot{f}(t)\times\Delta t + \epsilon\times\Delta t$ for some infinitesimal $\epsilon$.
\end{theorem}
\begin{proof}
  See~\cite{keisler2013elementary} Chapter-2.
\end{proof}

We are now ready to derive the chain rule in the hyperreals, which will
form the basis of our numerical simulation algorithm.

\begin{theorem}[Chain Rule]
  \label{thm:chain}
  Let $f$ and $G$ be two real functions and we define a new function
  $g(t)$ as $g(t) {\overset {\underset {\mathrm {def} }{}}{=}} G(f(t))$ for some
  independent real variable $t$. Let $x = f(t)$, $y = g(t) = G(x)$, by
  substituting $x\ \mathtt{for}\ f(t)\ in\ G$. At any value of $t$ where
  the derivatives $\dot{f}(t)$ and $\dot{G}(f(t))$ exist. Then
  \[
    \frac{\Delta y}{\Delta t} = \frac{\Delta y}{\Delta x} \times \frac{\Delta x}{\Delta t}
  \]
  \noindent
  where $\Delta y$, $\Delta x$, $\Delta t$ are infinitesimals, and
  $\Delta t \neq 0$,\ $\Delta x \neq 0$.
\end{theorem}
\begin{proof}
  From Increment Theorem~\ref{thm:inc}, we have
  $\Delta y = \dot{G}(x)\Delta x + \epsilon \times \Delta x$. Dividing both sides by
  $\Delta t$, we have
  \mbox{$\frac{\Delta y}{\Delta t} = \frac{\dot{G}(x)\Delta x + \epsilon \times \Delta x}{\Delta t}$}.
  Therefore,
  \mbox{$\frac{\Delta y}{\Delta t} = (\dot{G}(x)+ \epsilon)\times\frac{\Delta x}{\Delta t}$}. Finally,
  again applying the increment Theorem~\ref{thm:inc}, we have
  \mbox{$\frac{\Delta y}{\Delta t} = \frac{\Delta y}{\Delta x}\times\frac{\Delta x}{\Delta t}$}.
\end{proof}


\section{\acf{HA} and \acf{FA} --- syntax and semantics}
\label{sec:ha-syntax-semantics}

We convert a standard \acf{HA} to a \acf{FA} for execution
(Section~\ref{sec:prop-simul-techn}). This translation should be sound,
i.e., for every execution trace of the \ac{HA} there should be an
equivalent execution trace of the \ac{FA}. In this section we formalize
the soundness theorem of \ac{HA} and \ac{FA} executions.

We start with giving the syntax and execution trace semantics of the
\ac{HA} followed by the syntax and execution trace semantics of the
\ac{FA}. In this process we will use the steering wheel running example
(Section~\ref{sec:runn-example-probl}) to formalize the translation from
the \ac{HA} to the \ac{FA}.

\begin{definition}[\textbf{Syntax of \ac{HA}}]
  \label{def:ha}
  A \ac{HA} $\mathcal{H}$ is a tuple $\langle L, X, Init, f, Inv, E, G, R \rangle$, where:
  \begin{itemize}
    \setlength\itemsep{1pt}
  \item $L$ is a set of discrete locations with domain $\mathbf{L}$.
  \item $X$ is a finite collection of continuous variables, with its
    domain represented as $\mathbf{X} = \mathbb{R}^n$. 
  \item $Init \subseteq \{l_0\} \times \mathbf{X}$
    such that there is exactly one $l_0 \in \mathbf{L}$, is the
    singleton initial location.
  \item $f : \mathbf{L} \times \mathbf{X} \rightarrow \bm{X}$ is a vector field. Each
    flow in the vector field is only of the form $f(x(t))$.
  \item $h: \mathbf{L} \times \mathbf{X} \rightarrow \bm{X}$ is an update function.
  \item $Inv: \mathbf{L} \rightarrow 2^{\mathbf{X} }$ assigns
    to each $l \in \mathbf{L}$ an invariant set.
  \item $E \subset \mathbf{L} \times \mathbf{L} $ is a collection of discrete edges.
  \item $G : E \rightarrow 2^{\mathbf{X}}$ assigns to
    each $e = (l, l') \in E$ a guard.
  \item $R : E \times \mathbf{X} \rightarrow \mathbf{X}$ assigns to each
    $e = (l,l') \in E$, $x \in \mathbf{X}$ a reset relation.
  \end{itemize}
\end{definition}


\begin{definition}
  \label{def:guards}
  The edge guards $G$ are of the form:
  $g {\overset {\underset {\mathrm {def} }{}}{=}} x \bowtie \mathbb{Q}$. Where
  $x \in X$ is a continuous variable, $\bowtie\ \in \{\geq, \leq\}$.
\end{definition}

We can map the pictorial representation of the steering wheel \ac{HA} in
Figure~\ref{fig:stwheel} to the formal definition as follows:
\textcircled{1} $L=\{L1, L2\}$.\@ \textcircled{2} $X = \{x(t), y(t)\}$
and $\mathbf{X} = \mathbb{R}^{2}$.\textcircled{3}
$Init = \{L1\} \times \{\pi/2\} \times \{1\}$. Formally, the \acp{ODE} are
represented as vector flows
$f(L1, x(t), y(t)) {\overset {\underset {\mathrm {def} }{}}{=}}
{[\dot{x}(t), \dot{y}(t)]}' = {[0.1, 0]}'$.\@ \textcircled{4} An example
update function for location $L1$ is
$h(L1, x(t), y(t)) {\overset {\underset {\mathrm {def} }{}}{=}} {[x(t),
  y(t)]}' = {[x(t), \cos(x(t))]}'$. \@ \textcircled{5} Invariants are
used as \emph{fairness conditions} to enable an exit from any location
as soon as the invariants become false. They also restrict all
continuous variables to obey the invariant conditions, while the
execution remains in a given location i.e.\@ the invariant for the
location \textbf{L1} is $y(t) \leq -0.99$.\@ \textcircled{6} There are two
edges $e1 = (L1, L2)$ and $e2 = (L2, L1)$ connecting the two
locations.\@ \textcircled{6} An example edge guard
$G(e1) = y(t) \leq -0.99$. \textcircled{7} Finally, reset conditions like
the location updates ($h$) are given as
$R(e_{1}, x(t), y(t))\ {\overset {\underset {\mathrm {def} }{}}{=}}\
{[x'(t), y'(t)]}' := {[x(t), y(t)]}'$, where $x'(t)$, $y'(t)$ give the
updated values of the continuous variables when taking the
edge.\footnote{[]' represents vector transpose.}


\begin{definition}[\textbf{Semantics of \ac{HA}}]
  The semantics of any \ac{HA} is the set of all possible execution
  traces, where a valid execution trace is formalised below. An
  execution trace of \ac{HA} $\mathcal{H}$ is the tuple
  $x_{t} = (v(t), \bm{X}(t))$, if there exists a sequence of stopping
  times $T_{0} = 0 \leq T_{1} \leq T_{2} \leq T_{k},\ldots$ such that, for each
  $k \in \mathbb{N}$, and $T_{k} \in \mathbb{R}^{\geq 0}$ and
  $v(t) \in \bm{L}$, and finally, $\bm{X}(t)$ are the state variables of
  the \ac{HA}:
  \begin{itemize}

  \item $x_{0} = (v(0), \bm{X}(0))$, satisfies the initial condition
    $Init$.

  \item For $t \in [T_{k}, T_{k+1})$, $v(t) = v(T_{k})$ is constant and
    $\bm{X}(t)$ is a solution to the \ac{ODE} satisfying $flow(v(t))$.
    Moreover, $x(t) \in Inv(v(t))$.

  \item For \emph{some} outgoing edge of location $v(T_{k})$, i.e.,
    $\exists e \in E, s.t., e = (v(T_{k}), v(T_{k+1}))$: there, exists
    $T_{k+1} = T_{k} + S^{k}_{v(t)}$, where
    $S^{k}_{v(t)}= \min(t_{*}, t'_{*},\ldots)$. Here, any $t_{*},\ldots$ indicates
    the first time instants for \emph{all} outgoing edge guards for
    location $v(T_{k})$.
    
  \item Computing stop times $t_{*}$ requires that the \acp{ODE} evolve
    their respective continuous variables towards the satisfaction of
    edge guard predicates.

  \item If multiple outgoing edges, from any given location $v(T_{k})$,
    can be taken at the same instant, then we choose an outgoing edge
    randomly.

  \item The valuation $\bm{X}(T_{k+1})$ is governed by $R(e, x)$, where
    $e \in E = (v(T_{k}), v(T_{k+1}))$.
  \end{itemize}
  \label{def:hasemantics}
\end{definition}


\begin{definition}[\textbf{Syntax of \ac{FA}}]
  For a given \ac{HA} $\mathcal{H}$ the corresponding \acf{FA} $\mathcal{F}$ is a tuple:
  $\langle L_{f}, X_{f}, X^{n}_{f}, \Theta_{f}, Init_{f}, \Omega_{f}, Inv_{f}, E_{f},
  G_{f}, R_{f} \rangle$. Where:
  \begin{itemize}
  \item $L_{f} = L$; \ac{FA} has the same number of locations.
  \item $X_{f} = X$; all continuous state variables are carried over to
    the \ac{FA}.
  \item $Init_{f} = Init$; the initial of state variables is carried over
    from the \ac{HA}.
  \item
    $X^{n}_{f}: \bm{L}_{f} \times \bm{X}_{f} \rightarrow \mathbb{R}^{n}$, \textrm{such that}
    $x^{n}_{f}(l, x_{f}(t)) = \frac{(x_{f}(t) -
      x_{f})}{\max(\mathrm{range}(l, x_{f}(t)))}\ \forall x \in X^{n}_{f}, \forall l \in
    \mathbf{L}_{f}$, we introduce the normalized version of all
    continuous variables, where $x_{f}$ is the initial value of the
    state variable when it enters location $l$. Furthermore,
    $\mathrm{range}(l, x_{f}(t))$ is the range of all values that the
    variable $x_{f}(t)$ can take in the location $l$ during execution.
  \item $\Theta_{f} \in \mathbb{R}^{n}$ are the vector angle projections for each state
    variable on the unit circle (c.f.
    Section~\ref{sec:prop-simul-techn}).
  \item
    $\Omega_{f} : \bm{L} \times \bm{X^n_{f}} \rightarrow \bm{\Theta}^{n}_{f}$; are the angular
    velocity (angular frequency) of the individual vectors projected
    onto the unit circle.
  \item $Inv_{f} = Inv$, the invariant set for each location are carried
    over from the \ac{HA}.
  \item $E_{f} = E$, we have the same number of edges in the \ac{FA} as
    the \ac{HA}.
  \item $G_{f} : E \times 2^{\bm{X}_{f}} \rightarrow \Theta_{f}$ translates the guards from
    the \ac{HA} to an angle of the vector in the \ac{FA}.
  \item
    $R_{f} : E \times \{\bm{X}_{f} \cup \bm{\Theta}_{f} \cup \bm{X}^{n}_{f}\} \rightarrow
    \{\bm{X}_{f} \cup \bm{\Theta}_{f} \cup \bm{X}^{n}_{f}\}$ reset relations update
    all the variables in the \ac{FA}.
  \end{itemize}
  \label{def:fa}
\end{definition}

For the steering wheel running example \ac{HA} (c.f.
Definition~\ref{def:ha}), we can translate it to the \ac{FA} as follows:

\begin{align*}
  L_{f} {\overset {\underset {\mathrm {def} }{}}{=}} \{{L1}, {L2}\},\ 
  X_{f} {\overset {\underset {\mathrm {def} }{}}{=}} [x_{f}(t), y_{f}(t)]' =
  {[x(t), y(t)]}' \\
  X^{n}_{f}({L1},x_{f}(t), y_{f}(t)) {\overset {\underset {\mathrm {def} }{}}{=}}
  {[x^{n}_{f}(t), y^{n}_{f}(t)]}' = {[(x(t)-\pi/2)/0.99, y(t)-\cos(\pi/2)]}' \\
  \Theta_{f} {\overset {\underset {\mathrm {def} }{}}{=}} [\theta_{x}(t), \theta_{y}(t)]',\ 
  Init_{f}{\overset {\underset {\mathrm {def} }{}}{=}} \{L1\} \times \{\pi/2\} \times \{1\}\\
  \Omega_{f}({L1}) {\overset {\underset {\mathrm {def} }{}}{=}}
  [\dot{\theta}_{x}(t), \dot{\theta}_{y}(t)]' = [D(\dot{x}^{n}(t)), 0]',\ 
  Inv_{f}({L1}){\overset {\underset {\mathrm {def} }{}}{=}} \{\neg(y(t) \leq -0.99)\},\ 
  E_{f}{\overset {\underset {\mathrm {def} }{}}{=}} \{({L1}, {L2}), ({L2}, {L1})\} \\
  G_{f}({L1}, {L2}, x_{f}(t), y_{f}(t))
  {\overset {\underset {\mathrm {def} }{}}{=}}
  [\theta_{x}(t)] = [\arcsin(\frac{\arccos(-0.99)-\pi/2}{\max(range(x_{f}(t)))})] \\
  R_{f}({L1}, {L2}) {\overset {\underset {\mathrm {def}
  }{}}{=}} {[x_{f}'(t), y'_{f}(t), \theta'_{x}(t), \theta'_{y}(t), x'^{n}_{f}(t), y'^{n}_{f}(t)]}^{'} =
  {[x_{f}(t), y_{f}(t),0, 0, 0, 0]}^{'}
\end{align*}

Note in particular that the normalization of variables ($X^{n}_{f}$) can
be carried out statically, without solving the \ac{HA}. In order to
normalize any variable, we subtract the initial value that the variable
will take when entering any location, and then divide it with maximum of
all the different values that the variable can take when taking outgoing
edges and invariants of the location. For example, in location ${L1}$,
the normalized variable $x_{f}^{n}(t)$ can be calculated by:
\textcircled{1} subtracting the initial value $\pi/2$. \textcircled{2} Now
since the state variable $x_{f}(t)$ goes from $\pi/2$ to
$abs(\arccos(-0.99))$ in location $L1$, we can divide the resultant
value with $\max([\pi/2, abs(\arccos(-0.99))])$. We can always recompute
the value of original state variable by just multiplying the normalized
variable with the level-crossing and adding the initial value. The
initial value of the variable when it enters any location and its range
for a given location are known statically, because of
Definition~\ref{def:guards}.

Finally, note that the \acp{ODE} evolving the angular frequency are a
function ($D(x^{n}_{f}(t))$) of the normalized variables. We will show
later on that
$D(x^{n}_{f}(t)) {\overset {\underset {\mathrm {def} }{}}{=}}
\frac{\dot{x}^{n}_{f}(t)}{\cos(\theta_{x}(t))}$ for any normalized variable
$x^{n}_{f}(t)$, where $\cos(\theta_{x}(t)) \neq 0$.

\begin{definition}[\textbf{Semantics of \ac{FA}}]
  An execution trace of \ac{FA} $\mathcal{F}$ is
  $x_{ft} = (v_{f}(t), \bm{X}_{f}(t))$, if there exists a sequence of
  stopping times $T_{f0} = 0 \leq T_{f1} \leq T_{f2} \leq T_{fk},\ldots$ such that,
  for each $k \in \mathbb{N}$, and $T_{fk} \in \mathbb{R}^{\geq 0}$, $v_{f}(t) \in \bm{L}_{f}$:
  \begin{itemize}

  \item $x_{f0} = (v_{f}(0), \bm{X}_{f}(0))$, satisfies the initial
    condition $Init_{f}$.

  \item For $t \in [T_{fk}, T_{fk+1})$, $v_{f}(t) = v_{f}(T_{fk})$ is
    constant and $\bm{X}_{f}(t)$ is a solution to the \ac{ODE}
    satisfying $\Omega_{f}(v_{f}(t))$. Moreover,
    $x_{f}(t) \in Inv_{f}(v_{f}(t))$.

  \item For \emph{some} outgoing edge of location $v_{f}(T_{fk})$, i.e.,
    $\exists e_{f} \in E_{f}, s.t., e_{f} = (v_{f}(T_{fk}), v_{f}(T_{fk+1}))$:
    there, exists $T_{fk+1} = T_{fk} + S^{k}_{fv_{f}(t)}$, where
    $S^{k}_{fv_{f}(t)}= \min(t_{f*}, t_{f*}',\ldots)$. Here, any
    $t_{f*},\ldots$ indicates the first time instants for \emph{all} outgoing
    edge guards for location $v_{f}(T_{fk})$.
    
  \item Computing stop times $t_{f*}$ requires that the angular
    frequency evolve their respective continuous variables towards the
    satisfaction of edge guard predicates.

  \item If multiple outgoing edges, from any given location
    $v_{f}(T_{fk})$, can be taken at the same instant, then we choose an
    outgoing edge randomly.

  \item The valuation $\bm{X}_{f}(T_{fk+1})$ is governed by
    $R_{f}(e_{f}, x_{f})$, where
    $e_{f} \in E_{f} = (v_{f}(T_{fk}), v_{f}(T_{fk+1}))$.
  \end{itemize}
  \label{def:fasemantics}
\end{definition}

\subsection{Equivalence of the semantics of \acf{HA} and \acf{FA}}
\label{sec:equiv-acha-acfa}

In this section we give the proof of soundness of our translation scheme
from \ac{HA} to \ac{FA}. The primary idea is to give an equivalence
relation between the traces of \ac{HA} and \ac{FA}. The equivalence
relation ($\approxeq$) states that: \textcircled{1} the initial, at time
$t=0$, trace of both the executions are equal. \textcircled{2} Both
\ac{FA} and \ac{HA} take an outgoing edge from any given location at the
same time. \textcircled{3} When they both take an outgoing edge; the
source and destination locations are the same, along with the values of
the state variables. We give the formal statement along with the
associated proof next.

\begin{theorem} Let the trace of a well-formed \ac{HA} $\mathcal{H}$ be given by
  $\llbracket \mathcal{H} \rrbracket$ according to
  Definition~\ref{def:hasemantics}. Similarly, the trace of a translated
  \ac{FA} $\mathcal{F}$, be given by $\llbracket \mathcal{F} \rrbracket$ according to
  Definition~\ref{def:fasemantics}.
  $\llbracket \mathcal{H} \rrbracket \approxeq \llbracket \mathcal{F} \rrbracket$, if
  \begin{enumerate}
  \item $(v(0), \bm{X}(0)) = (v_{f}(0), \bm{X}_{f}(0))$
  \item For all stop times
    $T_{k}\ \mathrm{of}\ \llbracket \mathcal{H} \rrbracket$, for all stop times
    $T_{fk'}\ \mathrm{of}\ \llbracket \mathcal{F} \rrbracket$,
    $T_{k} = T_{fk'}$ and $k = k'$.
  \item For all stop times
    $T_{k}\ \mathrm{of}\ \llbracket \mathcal{H} \rrbracket$, for all stop times
    $T_{fk'}\ \mathrm{of}\ \llbracket \mathcal{F} \rrbracket$,
    $v(T_{k}) = v_{f}(T_{fk'})$, and
    $\bm{X}(T_{k}) = \bm{X}_{f}(T_{fk'})$.
  \end{enumerate}
  \label{fahaeq}
\end{theorem}

\begin{proof}
  First let us assume that \ac{HA} and \ac{FA} are in the same location,
  such that $v(T_{k}) = v(T_{fk})$ and $T_{k} = T_{fk}$, and
  $\bm{X}(T_{k}) = \bm{X}_{f}(T_{fk})$ for some stop times
  $T_{k} = T_{fk} \geq 0$. Now using Lemma~\ref{lemma:eqstatevar}, we know
  that \underline{equality}:
  $x(t) = \max(\mathrm{range}(x(t))) \times \sin(\theta(t)) + x(T_{k}) = x_{f}(t),
  \forall t \in [T_{k}, T_{k+1})$, holds for any state variable $x(t)$ in
  \ac{HA} and $x_{f}(t)$ in \ac{FA}. Next, using
  Lemma~\ref{lemma:guardconverge}, we know that any state variable
  $x_{f}(t)$ \textbf{converges} to the value that satisfies any outgoing
  edge guard from location $v_{f}(t)$. Hence, we can get
  $S^{k}_{fv_{f}(t)}$, which is the minimum from amongst all stop times
  ($t_{f*}$) satisfying the outgoing edge guard from location
  $v_{f}(t)$. Consequently, from the equality above, we have the
  \underline{step}: $T_{k+1} = T_{fk+1} = T_{k} + S^{k}_{fv_{f}(t)}$ and
  $\bm{X}(T_{k+1}) = \bm{X}_{f}(T_{fk+1})$.

  Finally, the first point (1) in the statement of the theorem holds
  trivially by the definition of \ac{HA} and \ac{FA}
  ($Init_{f} = Init$). Then using induction on the step above we can
  show that points (2) and (3) hold. Hence, the equivalence relation
  $\llbracket\mathcal{H}\rrbracket\approxeq\llbracket\mathcal{F}\rrbracket$ holds.
\end{proof}

\begin{lemma}[Equivalence relation of state-variables]
  \label{lemma:eqstatevar}
  Let $x(t)$ be a state variable in \ac{HA}. Let $x_{f}(t)$ be the
  representative state variable in \ac{FA} with its normalization
  $x^{n}_{f}(t)$ (Definition~\ref{def:fa}). If
  $\Delta x^{n}_{f} = \Delta \sin$, then
  $x(t) = x_{f}(t) = \max(\mathrm{range}(x(t)))\times\sin(\theta_{x}(t))+x(T_{k})$
  for some stop time $T_{k}$ when the \ac{HA} and \ac{FA} enter the some
  location, such that, $v(T_{k}) = v_{f}(T_{fk})$, $T_{k} = T_{fk}$, for
  all $t \in [T_{k}, T_{k+1})$.
\end{lemma}
\begin{proof}
  Consider the vector representing the variable $x_{f}(t)$ on the unit
  circle (c.f. Section~\ref{sec:prop-simul-techn}). This vector has an
  angle of $\theta_{x}(t)$. It's value on the Y-axis of the unit circle is
  given by $\sin(\theta_{x}(t))$. For an infinitesimal change
  $\Delta \sin$; from Theorem~\ref{thm:chain}, we have
  $\frac{\Delta \sin}{\Delta t} = \frac{\Delta \sin}{\Delta \theta_{x}} \times \frac{\Delta \theta_{x}}{\Delta t}$.
  Then from assumption $\Delta x^{n}_{f} = \Delta \sin$, we have:
  \begin{align*}
    \frac{\Delta x^{n}_{f}}{\Delta t} &= \frac{\Delta \sin}{\Delta \theta_{x}} \times \frac{\Delta \theta_{x}}{\Delta t} \implies
                              (\dot{x}^{n}_{f}(t) +  \epsilon_{x}) = (\dot{\sin}(\theta_{x}) + \epsilon_{\theta_{x}}) \times (\dot{\theta_{x}} + \epsilon_{t}) (\mathrm{from\ Theorem~\ref{thm:inc}})\\
    \therefore (\dot{x}^{n}_{f}(t) +  \epsilon_{x}) &= \cos(\theta_{x}) \dot{\theta_{x}}(t) +
                                      (\cos(\theta_{x}) \epsilon_{t} + \dot{\theta_{x}}\epsilon_{\theta_{x}} + \epsilon_{\theta_{x}}\epsilon_{t}) (\mathrm{By\ Definitions~\ref{def:extensionp}\ and~\ref{def:transfer}})\\
    \therefore st(\dot{x}^{n}_{f}(t) +  \epsilon_{x}) &= st(\cos(\theta_{x}) \dot{\theta_{x}}(t) +
                                        (\cos(\theta_{x}) \epsilon_{t} + \dot{\theta_{x}}\epsilon_{\theta_{x}} + \epsilon_{\theta_{x}}\epsilon_{t})) (\mathrm{By\ Definition~\ref{def:st}})\\
    \therefore \dot{x}^{n}_{f}(t) &= \cos(\theta_{x}) \dot{\theta_{x}}(t) \implies
                           \int\dot{x}^{n}_{f}(t) dt = \int \cos(\theta_{x}) \dot{\theta_{x}}(t) dt \\
    \therefore x^{n}_{f}(t) &= \sin(\theta_{x}(t)) + C (\textrm{By\ u-substitution}) \numberthis \label{eq:xtheta}\\
    \therefore x^{n}_{f}(t) &= \sin(\theta_{x}(t)), \mathrm{C = 0, because\ x^{n}_{f}(t)\ is\ normalised\ with\ \theta_{x}(0) = x^{n}_{f}(0) = 0} \\
    \therefore x_{f}(t) &= \max(\mathrm{range}(x(t))) \times \sin(\theta_{x}(t)) + x(T_{fk}) (\mathrm{By\ Definition~\ref{def:fa}})\\
    \therefore x(t) &= \max(\mathrm{range}(x(t))) \times \sin(\theta_{x}(t)) + x(T_{k}) (\mathrm{By\ Definition~\ref{def:fa}\ and\ assumption})
  \end{align*}
\end{proof}

\vspace{-20pt}

\begin{lemma}[Convergence to the guard]
  \label{lemma:guardconverge}
  Consider some state variable $x(t)$ in location $v(t)$ in a well
  formed \ac{HA}, where the \acp{ODE} in the location evolve the \ac{HA}
  towards one of its outgoing guards, $\mathcal{H}$ for
  $\forall t \in [T_{k}, T_{k+1})$. Consider, without loss of generality, some
  out going edge ($e$) guard from location $v(t)$, such that
  $x(T_{k+1})\ \mathrm{satisfies}\ G(e)$, where $T_{k+1}$ is the first
  time where the outgoing edge guard condition holds. Then
  $\int_{t}^{t+\Delta t}\dot{x}(t) d\tau\ \mathrm{satisfies}\ G(e)$ and
  $t + \Delta t = T_{k+1}$.
\end{lemma}
\begin{proof}
  We know that $x(T_{k+1})$ satisfies the outgoing edge guard condition.
  Hence, $\Delta x = x(T_{k+1}) - x(t)$ is the remaining increment (c.f.
  Definition~\ref{def:inc}) to reach the required level crossing. There
  exists a $\Delta x_{f} = \Delta x$ from Definition~\ref{def:fa}. Hence,
  $\Delta x^{n}_{f} = \frac{\Delta x_{f}}{\max(\mathrm{range}(x_{f}(t)))}$ from
  the definition of normalized variables in Definition~\ref{def:fa}.
  However, $\Delta x^{n}_{f} = \Delta \sin$ from Lemma~\ref{lemma:eqstatevar}.
  Hence, we have:
  $\Delta\theta_{x} = \theta_{x}(T_{k+1}) - \theta_{x}(t) = \arcsin(x^{n}_{f}(T_{k+1})) -
  \arcsin(x^{n}_{f}(t))$.

  In the general case, $\arcsin$ gives two possible angles (except at
  $\pi/2$ and $3\pi/2$). Figure~\ref{fig:arcsin} shows the scenario. At time
  $t$ we have two possible vectors marked 1 and 2, showing the position
  of vectors corresponding to $\arcsin(x^{n}_{f}(t))$. Similarly, we
  have two possible vectors marked $1'$ and $2'$ corresponding to the
  value of $\arcsin(x^{n}_{f}(T_{k+1}))$.

  Let us consider only vector 1, moving to vector $1'$ or $2'$ in order
  to satisfy the level crossing (same logic applies to vector 2). We
  have \textit{four} possible $\Delta\theta_{x}$ corresponding to the clockwise
  and anti clockwise rotation of the vector. The clockwise rotation of
  vector 1 to $1'$ and $2'$ are shown in Figure~\ref{fig:arcsin} as
  $-\Delta\theta_{1}$ and $-\Delta\theta_{2}$, respectively. The anti clockwise rotation of
  vector 1 to $2'$ and $1'$ would be $+\Delta\theta_{1}$ and
  $+\Delta\theta_{2}$, respectively. Hence,
  $\Delta\theta_{x} \in \{-\Delta\theta_{1}, -\Delta\theta_{2}, +\Delta\theta_{1}, +\Delta\theta_{2}\}$. Moreover, we also
  know that $0 < abs(\Delta\theta_{x}) < 2\pi$. Since, any value of
  $abs(\Delta\theta) \geq 2\pi$ can be subtracted by $2\pi$ to consider only a single
  rotation of the vector in the unit circle. If the
  $abs(\Delta\theta_{x}) = 0 = 2\pi$, then $\Delta\sin = 0$, which means
  $\Delta x^{n}_{f} = \Delta x_{f} = \Delta x = 0$. This can only happen if we have
  already reached the level crossing in a well formed \ac{HA}.

  Now we will relate the change in $\theta_{x}$ ($\Delta\theta_{x}$) to time step
  $\Delta t$. From Equation~(\ref{eq:xtheta}) in
  Lemma~\ref{lemma:eqstatevar}, we have:
  $\dot{\theta}_{x}(t) = \frac{\dot{x}^{n}_{f}(t)}{\cos(\theta(t))}$. Hence,
  applying the increment theorem (Theorem~\ref{thm:inc}) to
  $\Delta\theta_{x}$ and substituting $\dot{\theta}_{x}(t)$ we have:
  $\Delta\theta_{x} = \left (\frac{\dot{x}^{n}_{f}(t)}{\cos(\theta(t))} + \epsilon_{t}\right
  )\Delta t$. Increment $\Delta t$ can only be non-negative, since time can only
  proceed forward. Furthermore, $\Delta t = 0$, when
  $\Delta\theta_{x} = 0$, which entails
  $\Delta\sin = \Delta x^{n}_{f} = 0$. This is the case when we have already
  reached the level crossing, as mentioned before. Now let us consider
  the remaining cases for
  $\dot{\theta}_{x}(t) = \frac{\dot{x}^{n}_{f}(t)}{\cos(\theta(t))}$.

  \begin{enumerate}
  \item $\dot{\theta}_{x}(t) < 0$: then $\Delta \theta$ has to be less than zero, for
    $\Delta t$ to be greater than zero. In such a case, we choose
    $\Delta\theta_{x} = \min\{abs(-\Delta\theta_{1}), abs(-\Delta\theta_{2})\}$. Hence,
    $t + \Delta t = T_{k+1}$ is the first instant when
    $x^{n}_{f}(T_{k+1})$ satisfies the guard condition. Following from
    Lemma~\ref{lemma:eqstatevar},
    $\int_{t}^{t+\Delta t}\dot{x}(t) d\tau\ \mathrm{satisfies}\ G(e)$ and
    $t + \Delta t = T_{k+1}$.
  \item $\dot{\theta}_{x}(t) > 0$: then $\Delta \theta$ has to be greater than zero,
    for $\Delta t$ to be greater than zero. In such a case, we choose
    $\Delta\theta_{x} = \min\{+\Delta\theta_{1}, +\Delta\theta_{2}\}$. Hence,
    $t + \Delta t = T_{k+1}$ is the first instant when
    $x^{n}_{f}(T_{k+1})$ satisfies the guard condition. Following from
    Lemma~\ref{lemma:eqstatevar},
    $\int_{t}^{t+\Delta t}\dot{x}(t) d\tau\ \mathrm{satisfies}\ G(e)$ and
    $t + \Delta t = T_{k+1}$.
  \end{enumerate}

  Finally, if we choose $\Delta t' < \Delta t$, where
  $\Delta t = T_{k+1} - t$, then by induction we always have another
  $\Delta t''$, such that
  $\int_{t+\Delta t'}^{t+\Delta t''}\dot{x}(t) d\tau\ \mathrm{satisfies}\ G(e)$ and
  $t+\Delta t' + \Delta t'' = T_{k+1}$.
\end{proof}

\begin{wrapfigure}{l}{0.45\textwidth}
  \centering
  \includegraphics[scale=1.5]{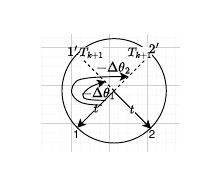}
  \caption{Level crossings and current vector position projected to the
    unit circle}
  \label{fig:arcsin}
\end{wrapfigure}

Now that we have shown that the \ac{FA} can simulate an \ac{HA}, we can
translate the \ac{HA} to a \ac{FA} for simulation. In the next section
we give the simulation algorithm for \ac{HA}.


\section{An efficient simulation algorithm }
\label{sec:an-effic-simul}

We take a top-down approach to explain the proposed numerical simulation
algorithm. Pseudo-code for executing an equivalent \ac{FA} translated
from a \ac{HA} is shown in Algorithm~\ref{alg:fa}. The algorithm takes
as input a \ac{HA}. The output is the trace $\mathcal{X}$ as defined in
Definition~\ref{def:fasemantics}.

A input \ac{HA} $\mathcal{H}$ is first translated into a \ac{FA}
$\mathcal{F}$ on line~\ref{line:convert} as described in Definition~\ref{def:fa}.
A global clock variable $T_{c}$ is initialised on line~\ref{line:clock}.
Next, a while loop executes, computing the trace of the \ac{FA} and
incrementing the clock ($T_{c}$) until user specified time $T_{max}$ is
reached (lines~\ref{line:while_start}-\ref{line:while_end}).

\begin{algorithm}
  \caption{Pseudo-Code for simulation of \ac{FA}}
  \begin{footnotesize}
    \begin{algorithmic}[1]
      \Require \ac{HA} $\mathcal{H}$ \Ensure Output Trace $\mathcal{X}$ of
      $\mathcal{F}$ \State
      $\mathcal{F} \leftarrow \mathtt{convertToFA}(\mathcal{H}$); \label{line:convert} \State
      $ T_{c} \leftarrow 0$ \Comment{Clock initialisation}; \label{line:clock}
      \While{$T_{c} \leq T_{max}$} \label{line:while_start} \State
      $\mathcal{X} \leftarrow \mathtt{Record}(T_{c}, \llbracket \mathcal{F} \rrbracket)$; \Comment{Record trace
        $\llbracket \mathcal{F} \rrbracket$ following
        Definition~\ref{def:fasemantics}}; \label{line:record}
      \If{$\mathcal{F}.\mathtt{guardEnabled}()$} \label{line:guard_start}
      \State $\mathcal{F}$.\texttt{switch}();\Comment{Location switch
        transition};\label{line:interloc} \State
      \textbf{continue};\Comment{Global clock does not
        increment} \label{line:intercont} \EndIf \label{line:guard_end}
      \State $t_{f*} \leftarrow \emptyset
      $; \label{line:intra_start}\Comment{Intra-location transition};
      \ForAll{$outgoing\ edges'\ guard\ (g_{f}) in\ the\ location$}
      \State $t_{f*}.\mathtt{append}(\mathcal{F}.\mathtt{computeDelta}(T_{c}, g_{f}))$;
      \label{line:compute_delta}
      \EndFor
      \State$S_{ft}\leftarrow \min(t_{f*})$;\label{line:intra_min} \State
      $\mathcal{F}.\mathrm{executeIntra}(S_{ft})$; \label{line:execintra} \State
      $T_{c} \leftarrow T_{c} + S_{ft}$ \Comment{Increment global
        clock}\label{line:intra_end} \EndWhile \label{line:while_end}
      \State \Return $\mathcal{X}$
    \end{algorithmic}
  \end{footnotesize}
  \label{alg:fa}
\end{algorithm}

The algorithm starts with recording the current state of the state
variables in the \ac{FA} (line~\ref{line:record}). Next, the algorithm
checks if any out going edge guards from the current location are
already enabled (line~\ref{line:guard_start}). If so, a switch is made
to the next location with the state (and associated) variables reset
according to the reset conditions of the guard
(line~\ref{line:interloc}). The loop continues without incrementing the
global clock ($T_{c}$) in such a scenario (line~\ref{line:intercont}).

If no outgoing edge guard is enabled, then a \textit{so called} intra
location transition is taken
(lines~\ref{line:intra_start}---\ref{line:intra_end}). First, the step
size $t_{*}$ is computed using function \texttt{computeDelta}, which
follows Lemma~\ref{lemma:guardconverge}, for each outgoing edge guard
from the current location, and is described next. The minimum ($S_{ft}$)
from amongst all these computed time steps is used as the time step to
evolve the \acp{ODE} in the current location
(line~\ref{line:execintra}). Finally, the global clock ($T_{c}$) is
incremented with $S_{ft}$ and the loop continues.

\subsection{Step Size Calculation Algorithm}
\label{ssec:step-size-alg}

The time step size calculation pseudo-code is shown in
Algorithm~\ref{alg:guard_step}. The algorithm takes as input the \ac{FA}
$\mathcal{F}$, along with the outgoing edge guard and the current global
simulation clock time $T_{c}$. Next, using
Lemma~\ref{lemma:guardconverge}, the maximum change in the vector angle
($\Delta\theta_{x}$) for some variable $x_{f}(t)$ for the given guard is
calculated along with the associated time step $\Delta t$
(lines~\ref{ll:1}---\ref{ll:4}). If the maximum change in the angle
$\Delta\theta_{x}$ (c.f. Lemma~\ref{lemma:guardconverge}) is greater than some
user specified change $\Delta\theta_{u}$, then the maximum change is reset to user
specified change.

From Lemma~\ref{lemma:guardconverge} we have
$\Delta\theta_{x} = \left (\frac{\dot{x}^{n}_{f}(t)}{\cos(\theta(t))} + \epsilon_{t}\right )\Delta
t$. Hence, computation of time step $\Delta t$ includes some infinitesimal
error $\epsilon_{t}$. In order to bound this error, the user specifies a
maximum error bound $\epsilon_{u}$. Two different normalized variable values at
time $T_{c} + \Delta t$ and $T_{c} + \Delta t/2$ are calculated for the computed
$\Delta t$ (Algorithm~\ref{alg:guard_step}, lines~\ref{ll:6}---\ref{ll:7}).
While the absolute difference between these two values
(line~\ref{ll:10}) is greater than the user specified error bound
$\epsilon_{u}$, we keep on halving the time step. Once the user specified error
bound is met, the resultant time step $\Delta t$ is returned
(line~\ref{ll:12}).

It is worth noting that since the time step is computed using the
Lemma~\ref{lemma:guardconverge} and the loop bounding the $\Delta t$ error
only ever reduces the time step, we are guaranteed to converge to the
level crossing without missing it as described in
Lemma~\ref{lemma:guardconverge}. We show this experimentally in the next
section.


\begin{algorithm}
  \caption{Compute time step size}
  \begin{footnotesize}
    \begin{algorithmic}[1]
      \Require current time $T_{c}$, angle of the guard $g_f(T)$
      \Ensure Step size $\Delta t$
      
      \Function{computeDelta}{$T$, $g_f(T)$} 
      \If{$\Delta\theta_{x} > \Delta\theta_{u}$} \label{ll:1}
      \State $\Delta\theta_{x} = \Delta\theta_{u}$ \Comment{Upper bound of change specified
        by user} \label{ll:2}
      \EndIf \label{ll:3}
      \State $\Delta t \leftarrow \frac{\Delta \theta_{x} 
        \cos(\theta_{x}(T))}{\dot{x}^{n}_{f}(T)}$
      \Comment{Compute $\Delta t$ using Lemma~\ref{lemma:guardconverge}} \label{ll:4}
      \Do \label{ll:5}
      \State $x^{n}_{fc} = \sin(\theta_{x}(T_{c} + \Delta t))$; \Comment{Compute
        normalized state variable using Lemma~\ref{lemma:eqstatevar}} \label{ll:6}
      \State $x^{n'}_{fc} = \sin(\theta_{x}(T_{c} + \Delta t/2))$; \Comment{Compute
        normalized state, for $\Delta t/2$ variable using
        Lemma~\ref{lemma:eqstatevar}} \label{ll:7}
      \State $\Delta t \leftarrow \Delta t/2$ \Comment{Half the time step} \label{ll:8}
      \doWhile{ \label{ll:9}
        $abs(x^{n}_{fc} - x^{n'}_{fc}) \leq \epsilon_{u}$ \Comment{Compute until
          user specified error bound} \label{ll:10}
      } \label{ll:11}
      \State \textbf{return} $\Delta t$ \label{ll:12}
      \EndFunction
    \end{algorithmic}
  \end{footnotesize}
  \label{alg:guard_step}
\end{algorithm}

\section{Experimental results}
\label{sec:experimental-results}

In this section, we perform a quantitative comparison of the proposed
simulation algorithm for \ac{HA} vis-\'a-vis\\ Simulink/Stateflow\textregistered, which
is an industry grade simulation tool for hybrid systems. We use three published
benchmarks with varying complexity of \acp{ODE} and guards for
comparison. The first benchmark is the steering wheel system, which was
introduced as our motivating example in
Section~\ref{sec:runn-example-probl}.

The second benchmark is the robot control example presented
in~\cite{malik2020dynamic} where the robot moves in a ellipsoidal
trajectory, in a two dimensional cartresian plane. There are two
locations in this benchmark. In the initial location the \acp{ODE} are:
$\dot{x}(t) = 5\sin(a(t))$, and $\dot{y}(t) = 5\cos(a(t))$ and
$\dot{a}(t) = 0.9$, where $x(t)$ and $y(t)$ are the horizontal and
vertical position of the robot and $a(t)$ is the heading. The robot
stops when it hits the obstacle. The position of the obstacle is
expressed as the guard constraint $y(t) \geq 12x^{2}(t) - 54x(t) + 65$.
Once this outgoing guard holds, a transition is made to the second
location, where evolution of the state variables stops and the system
halts, indicating that the robot stops moving.

The third benchmark is the modified water heating example
from~\cite{raskin2005introduction}. This benchmark has three locations:
\texttt{S0}, \texttt{ON}, and \texttt{OFF}, respectively. There are two
state variables $temp(t)$ and $time(t)$ indicating the temperature of
the water and system clock, respectively. The temperature of water is
initially set to 30 degrees and the \acp{ODE} in the location
\texttt{S0} are $\dot{time}(t) = 1$ and $\dot{temp}(t) = 0$. When the
outgoing edge guard from location \texttt{S0}: $timer(t) \geq 5$ holds; a
discrete transition to the \texttt{ON} location is made. In the
\texttt{ON} location; the \acp{ODE} for the state variables are
$\dot{time}(t) = 1$ and $\dot{temp}(t) = 0.075\times(150-temp(t))$. Once the
temperature of water reaches 100 degrees, it makes a transition to the
\texttt{OFF} location where all continuous variables stop evolving and
the model halts execution.

\subsection{Experimental Set-up}
\label{sec:experimental-set-up}

All our experiments are carried out on an Apple M1 Pro CPU executing at
3.2GHZ with 16GB memory. The Matlab release we use for the
Simulink/Stateflow\textregistered\ models is 2022a. The \ac{FA} models are implemented
and simulated in C. For all our Simulink/Stateflow\textregistered\ benchmark models,
we use the ODE45 solver. In Simulink/Stateflow\textregistered, the step size is in the
time domain, whereas in \ac{FA}, the step size is the increment of the
angle projected onto the unit circle for each continuous variable. We
set various maximum step sizes for each model and compare the results.
We use a high-resolution ODE45 solver with the maximum step of $1e^{-4}$ as the reference benchmark trace for all our case studies. We measure the simulation precision, which is measured by the correlation
coefficient to the reference model. 
The simulation efficiency is measured by the number of steps and the
execution time (in $\sec$) given identical maximum simulation time (maximum).

\subsection{Results}
\label{sec:results}

\begin{figure}[tbh]
	\centering
	
	\begin{subfigure}[b]{0.23\textwidth}
		
		\begin{tikzpicture}
			\begin{axis}[
				width=4cm, height=3.5cm,
                    domain=0:50,
                    samples=20,
				xlabel={\small Time (s)},
				ylabel={ \tiny $\cos(x(t))$},
				ylabel style={at={(axis description cs:0.2,.5)}, anchor=south},
				xmin=0, xmax=50,
				ymin=-1, ymax=1,
				]
				\addplot[mark=*, draw=black] 
				table [x=time_s_auto, y=out_s_auto, col sep=comma] {./data/case_one.csv};
			\end{axis}
		\end{tikzpicture}
		\caption{\scriptsize Simulink max step = auto}
		\label{fig:sl_step_auto}
	\end{subfigure}
	\hspace{0.03\textwidth}
	\begin{subfigure}[b]{0.23\textwidth}
		\begin{tikzpicture}
			\begin{axis}[
				width=4cm, height=3.5cm,
                    domain=0:50,
                    samples=20,
				xlabel={\small Time (s)}, 
				ylabel={ \tiny $\cos(x(t))$},
				ylabel style={at={(axis description cs:0.2,.5)}, anchor=south},
				xmin=0, xmax=50,
				ymin=-1, ymax=1,
				]
				\addplot[mark=*, draw=black] 
				table [x=time_s_0_1, y=out_s_0_1, col sep=comma] {./data/case_one.csv};
			\end{axis}
		\end{tikzpicture}
		\caption{\scriptsize Simulink max step = 0.1s}
		\label{fig:sl_step_0.1}
	\end{subfigure}
	\hspace{0.03\textwidth}
	\begin{subfigure}[b]{0.23\textwidth}
		\begin{tikzpicture}
			\begin{axis}[
				width=4cm, height=3.5cm,
                    domain=0:50,
                    samples=20,
				xlabel={\small Time (s)}, 
				ylabel={ \tiny $\cos(x(t))$},
				ylabel style={at={(axis description cs:0.2,.5)}, anchor=south},
				xmin=0, xmax=50,
				ymin=-1, ymax=1,
				]
				\addplot[mark=*, draw=black] 
				table [x=time_s_0.01, y=out_s_0.01, col sep=comma] {./data/case_one.csv};
			\end{axis}
		\end{tikzpicture}
		\caption{\scriptsize Simulink max step = 0.01s}
		\label{fig:sl_step_0.01}
	\end{subfigure}
	
	\vspace{0.5em}
	
	\begin{subfigure}[b]{0.23\textwidth}
		\begin{tikzpicture}
			\begin{axis}[
				width=4cm, height=3.5cm,
                    domain=0:50,
                    samples=20,
				xlabel={\small Time (s)}, 
				xmin=0, xmax=50,
				ylabel={ \tiny $\cos(x(t))$},
				ylabel style={at={(axis description cs:0.2,.5)}, anchor=south},
				ymin=-1, ymax=1,
				]
				\addplot[mark=*, draw=black] 
				table [x=time_fa_pi_10, y=out_fa_pi_10, col sep=comma] {./data/case_one.csv};
			\end{axis}
		\end{tikzpicture}
		\caption{\scriptsize FA max step-angle = $\pi/10$}
		\label{fig:fa_step_pi10}
	\end{subfigure}
	\hspace{0.03\textwidth}
	\begin{subfigure}[b]{0.23\textwidth}
		\begin{tikzpicture}
			\begin{axis}[
				width=4cm, height=3.5cm,
                    domain=0:50,
                    samples=20,
				ylabel={ \tiny $\cos(x(t))$},
				ylabel style={at={(axis description cs:0.2,.5)}, anchor=south},
				xlabel={\small Time (s)}, 
				xmin=0, xmax=50,
				ymin=-1, ymax=1,
				]
				\addplot[mark=*, draw=black] 
				table [x=time_fa_pi_50, y=out_fa_pi_50, col sep=comma] {./data/case_one.csv};
			\end{axis}
		\end{tikzpicture}
		\caption{\scriptsize FA max step-angle = $\pi/50$}
		\label{fig:fa_step_pi50}
	\end{subfigure}
	\hspace{0.03\textwidth}
	\begin{subfigure}[b]{0.23\textwidth}
		\begin{tikzpicture}
			\begin{axis}[
				width=4cm, height=3.5cm,
                    domain=0:50,
                    samples=20,
				ylabel={ \tiny $\cos(x(t))$},
				ylabel style={at={(axis description cs:0.2,.5)}, anchor=south},
				xlabel={\small Time (s)}, 
				xmin=0, xmax=50,
				ymin=-1, ymax=1,
				]
				\addplot[mark=*, draw=black] 
				table [x=time_fa_pi_100, y=out_fa_pi_100, col sep=comma] {./data/case_one.csv};
			\end{axis}
		\end{tikzpicture}
		\caption{\scriptsize FA max step-angle = $\pi/100$}
		\label{fig:fa_step_pi100}
	\end{subfigure}
	
	\caption{\scriptsize Comparison of Simulink and FA step sizes and precision. 
		Y-axis in all plots shows $\cos(x(t))$ for the steering wheel benchmark.}
	\label{fig:stepsize}
\end{figure}
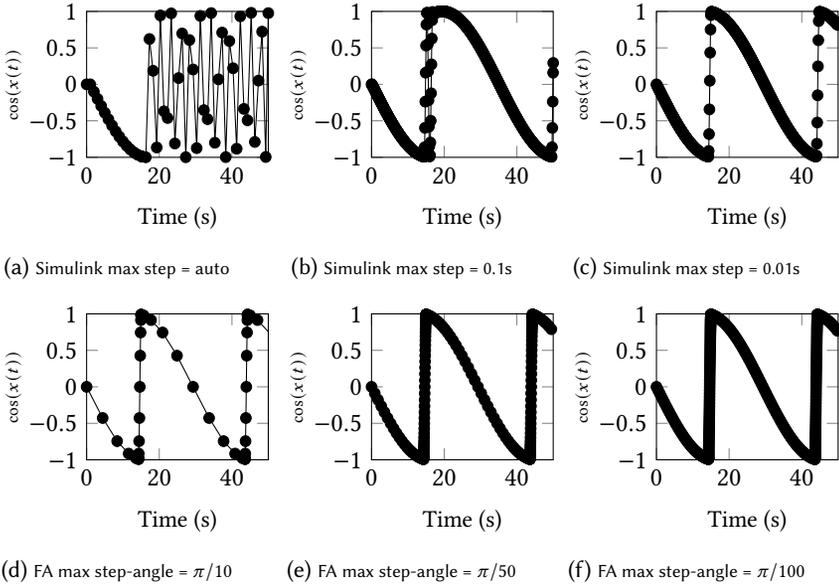

Figure~\ref{fig:stepsize} shows the steering wheel simulation result of
\ac{FA} and Simulink/Stateflow\textregistered. The plots in the top row
(Figures~\ref{fig:sl_step_auto},~\ref{fig:sl_step_0.1} and
~\ref{fig:sl_step_0.01}) are simulation results of Simulink/Stateflow\textregistered\
with various maximum step sizes. The plots in the bottom row
(Figures~\ref{fig:fa_step_pi10},~\ref{fig:fa_step_pi50} and
~\ref{fig:fa_step_pi100}) are the outputs generated using the proposed
technique with various maximum angle increments ($\Delta\theta_{u}$
Algorithm~\ref{alg:guard_step}). Each step is marked as a black circle
and the error tolerance is set to $1e^{-6}$ in both tools. In
Simulink/Stateflow\textregistered, there is an option to set the maximum time step
size to ``auto'', which allows the solver to select the maximum time
step size automatically. However, when this option is selected,
Simulink/Stateflow\textregistered\ cannot detect the level crossing, causing erroneous
behaviour (Figure~\ref{fig:sl_step_auto}). When the maximum time step
size is bounded to 0.1 $\sec$, Simulink/Stateflow\textregistered\ somewhat, detects
the level crossing. However, the stiff behaviour due to the sudden
change in ODE causes unstable location jump
(Figure~\ref{fig:sl_step_0.1}). Simulink/Stateflow\textregistered\ finally generates
the output similar to the reference with the maximum time step size
bounded to 0.01 $\sec$ (Figure~\ref{fig:sl_step_0.01}).

Figure~\ref{fig:fa_step_pi100} shows the simulation results of
\ac{FA} with the maximum angle increment bounded to $\pi/100$. It shows
that the level crossing is correctly detected.
Figures~\ref{fig:fa_step_pi50} and ~\ref{fig:fa_step_pi10} show the
simulation of \ac{FA} when the maximum angle increment is increased to
$\pi/50$ and $\pi/10$, respectively. As the figures show, the increased
maximum allowed angle increment reduces the number of steps while still
accurately detecting the level crossing. Notably, this detection occurs
with significantly fewer steps (Figure~\ref{fig:fa_step_pi10}) compared
to the Simulink/Stateflow\textregistered\ simulation (Figure~\ref{fig:sl_step_0.01}).

\begin{figure}[bth]
  \centering
  \scalebox{0.5}{
    \begin{tikzpicture}
      \centering
      \begin{axis}[
        view={-35}{25},
        xmin=0, xmax=10,
        domain=0:20,
        ymin=-6, ymax=6,
        zmin=0, zmax=8,
        scaled x ticks=false,
        xlabel={x}, 
        ylabel={y},
        zlabel={Time(s)},
        label style={font=\bfseries\boldmath},
        tick label style={font=\bfseries\boldmath},
        legend pos=outer north east,
        ]
        \addplot3 [ draw=red, line width=3pt] 
        table [x=x_ref, y=y_ref, z=time_ref, col sep=comma] {./data/robot_wrong.csv};
        \addlegendentry{Reference model}
        
        \addplot3 [ mark=x, draw=blue, thin] 
        table [x=x, y=y, z=time, col sep=comma] {./data/robot_wrong.csv};
        \addlegendentry{Simulink/Stateflow\textregistered\ 0.1 max timestep}
        \addplot3 [, mark=*, draw=green, line width=0.1pt, mark size=1.5pt, mark options={green}] 
        table [x=x_f_min, y=y_f_min, z=time_f_min, col sep=comma] {./data/robot_wrong.csv};
        \addlegendentry{\ac{FA} with $\pi/10$ max angle increment}
        
        \draw[->, ultra thick, black] (axis cs:3 ,0,1) -- (axis cs:2.144131, 4.184763,1.010667);
        \node at (axis cs:4,-1,1) {\textbf{Obstacle}};
      \end{axis}
    \end{tikzpicture}
  }
  \caption{\scriptsize Simulation of the second benchmark --- Robot Control}
  \label{fig:robot_wrong}
\end{figure}
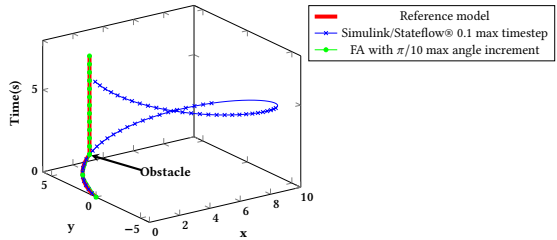

Figure~\ref{fig:robot_wrong} illustrates the simulation of the second
benchmark, which models the movement of a robot in the x-y coordinate
plane within a 3D space. The vertical axis is time. The red plot is the
reference model, obtained with a high-resolution ODE45 solver as
discussed in Section~\ref{sec:experimental-set-up}, which shows that the
robot stops moving upon reaching the obstacle boundary, defined by the
constraint $y(t) \geq 12x^{2}(t) - 54x(t) + 65$. If this guard condition is
met, it indicates that the robot has collided with the obstacle. The
simulation result of Simulink/Stateflow\textregistered\ is shown with the blue plot,
with the maximum step size set to 0.01 $\sec$ and the error tolerance to
$1e^{-2}$. The simulation steps, marked by blue crosses, indicate that
Simulink/Stateflow\textregistered\ fails to detect the guard, causing the robot to
continue moving \textit{through the obstacle} even after colliding with
it. The simulation results of \ac{FA} are shown with the green plot,
with simulation steps marked by green dots. The maximum angle increment
is set to $\pi/10$. The simulation of \ac{FA} correctly detects the guard
and stops the robot upon reaching the obstacle boundary.

Figure~\ref{fig:water_heating} shows the simulation results of the third
benchmark, water heating system. 
The guard predicate to detect the water temperature reaching 100 degrees
should be $temp(t) == 100$. However, Simulink/Stateflow\textregistered\ cannot handle
equality in guard conditions even with a very fine time step size
resolution of 0.0001 $\sec$. Figure~\ref{fig:water_heating_exact} shows
the simulation outputs of \ac{FA} and Simulink/Stateflow\textregistered\ with the
guard condition of $temp(t) == 100$. The green plot is the simulation of
Simulink/Stateflow\textregistered\ with the step size of 0.0001 $\sec$, and it fails
to detect the guard, causing the water temperature to keep increasing
over 100 degrees. On the other hand, \ac{FA} can accurately capture the
equality guard and make the correct transition, as shown in the blue
plot. Only when we change the guard condition to $temp(t) \geq 100$, can
Simulink/Stateflow\textregistered\ detect the guard as shown in
Figure~\ref{fig:water_heating_greater}. 

\begin{figure}[!htbp]
  \centering
   \hspace{-0.1\textwidth}
  \begin{subfigure}[b]{0.35\textwidth}
    \centering
    \begin{tikzpicture}
      \begin{axis}[
        width=4cm, height=4.2cm,
        xlabel={\tiny time}, 
        ylabel={\tiny temperature},
        ylabel style={at={(axis description cs:0.1,0.5)}, anchor=south},
        xmin=0, xmax=20,
        domain=0:20,
        ymin=10, ymax=120,
        scaled x ticks=false,
        label style={font=\bfseries\boldmath},
        tick label style={font=\bfseries\boldmath, font=\tiny},
        legend pos=outer north east,
        legend style={draw=none, fill=none, font=\scriptsize,
         cells={anchor=west}},
        legend style={align=left},
        ]
        \addplot[  draw=blue,  line width=1pt] 
        table [x=time_f, y=temp_f, col sep=comma] {./data/water.csv};
        \addlegendentry{\ac{FA} with $\pi/10$ \\ angle step}
        
        \addplot[draw=green,  line width=1pt] 
        table [x=time_equal, y=temp_equal, col sep=comma] {./data/water.csv};
        \addlegendentry{Simulink \\ 0.0001 max \\ step size}
      \end{axis}
      
    \end{tikzpicture}
    \caption{\scriptsize Guard: $temp(t) == 100$}
    \label{fig:water_heating_exact}
  \end{subfigure}%
  \hspace{0.1\textwidth}
  \begin{subfigure}[b]{0.35\textwidth}
    \centering
    \begin{tikzpicture}
      \begin{axis}[
        width=4cm, height=4.2cm,
        xlabel={\tiny time}, 
        ylabel={\tiny temperature},
        xmin=0, xmax=20,
        domain=0:20,
        ymin=10, ymax=120,
        scaled x ticks=false,
        label style={font=\bfseries\boldmath},
        tick label style={font=\bfseries\boldmath,font=\tiny},
        legend pos=outer north east,
        legend style={draw=none, fill=none, font=\scriptsize,
        ylabel style={at={(axis description cs:-0.1,0.5)}, anchor=south},	cells={anchor=west}},
        legend style={align=left},
        ]
        \addplot [ draw=red, line width=2pt] 
        table [x=time_ref, y=temp_ref, col sep=comma] {./data/water.csv};
        \addlegendentry{Reference model}
        
        \addplot [ mark=x, mark size = 1.2pt, draw=blue, opacity=1.0] 
        table [x=time_f, y=temp_f, col sep=comma] {./data/water.csv};
        \addlegendentry{\ac{FA} with $\pi/50$ \\ max angle increment}
        
        \addplot [ mark=star, mark size = 0.4pt, draw=green, opacity=0.8] 
        table [x=time_h, y=temp_h, col sep=comma] {./data/water.csv};
        \addlegendentry{Simulink \\ 0.1 max step size}
      \end{axis}
    \end{tikzpicture}
    \caption{\scriptsize Guard: $temp(t) \geq 100$}
    \label{fig:water_heating_greater}
  \end{subfigure}
  
  \caption{\scriptsize Equality guard condition for \ac{FA} and Simulink/Stateflow\textregistered\ }
  \label{fig:water_heating}
\end{figure}
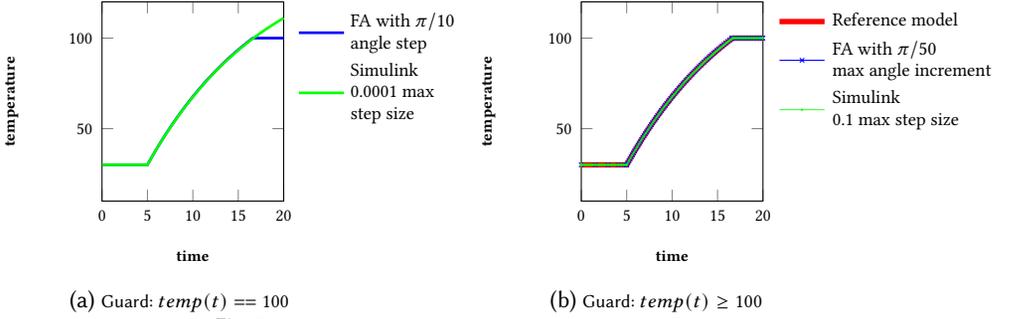

\begin{table}[bth]
  \scalebox{0.7}{
    \begin{tabular}{l c c c c c c r}
      \toprule
      &\multicolumn{2}{c}{Steering Wheel} & \multicolumn{2}{c}{Robot Control} & \multicolumn{2}{c}{Water heating}\\
      Tools (max step) & Accuracy & \# of Steps ($\sec$)&Accuracy & \# of Steps ($\sec$)& Accuracy&\# of Steps ($\sec$)\\
      \midrule
      
      \multicolumn{5}{l}{ Error tolerance = $1e^{-6}$} & \\
      \bottomrule
      Simulink(auto) & 0.110072092& 53(0.421585)& 0.57867301& 56(0.187890)& 0.91580121& 61(0.771031)\\
      Simulink/Stateflow\textregistered (0.1) & 0.45443298& 505(0.271238)& \textbf{>0.99999999}& \textbf{77(0.191814)}& 0.97567008& 209(0.753462)\\
      Simulink(0.01) & 0.97909321& 5008(0.377498)& >0.99999999& 706(0.252160))& 0.99746838& 2008(0.791454)\\
      Simulink(0.001) & \textbf{0.99979234}& \textbf{50007(0.950637)}& >0.99999999& 7005(0.582999)& \textbf{0.99979234}& \textbf{20006(1.023724)}\\
      
      FA ($\pi/10$) & \textbf{0.999969946}& \textbf{38(0.000969)}& \textbf{>0.99999999}& \textbf{24(0.000569)}& \textbf{>0.99999999}& \textbf{7802(0.785956)}\\ 
      FA ($\pi/50$) & 0.999970735& 190(0.003597)& >0.99999999& 31(0.000767)& >0.99999999& 7822(0.80612)\\ 
      FA ($\pi/100$) & 0.999970816& 380(0.007552)& >0.99999999& 39(0.001125)& >0.99999999& 7847(0.857779)\\ 
      FA ($\pi/150$) & 0.999970877& 570(0.007749)& >0.99999999&	44(0.001432)& >0.99999999& 7872(0.855666)\\ 
      
      \bottomrule
      \multicolumn{5}{l}{ Error tolerance = $1e^{-4}$} & \\
      \bottomrule
      
      Simulink(auto)& 0.110072146&	53(0.284067)& 0.516281396& 54(0.182498)& 0.915398777& 57(0.775229)\\
      Simulink(0.1)& 0.454437239& 505(0.342309)& \textbf{>0.99999999}& \textbf{74(0.177585)}& 0.975574341& 205(0.762533)\\
      Simulink(0.01)& 0.979093705& 5008(0.408343)& >0.99999999& 703(0.179463)& 0.99746749& 2005(0.776408)\\
      Simulink(0.001)& \textbf{0.999792386}& \textbf{50007(1.068943)}& >0.99999999& 7003(0.251786)& \textbf{0.999745788}& \textbf{20004(0.998744)}\\
      
      FA ($\pi/10$) & \textbf{0.999969946}& \textbf{38(0.000973)}& \textbf{>0.99999999}& \textbf{18(0.000422)}& \textbf{>0.99999999}& \textbf{813(0.05741)}\\ 
      FA ($\pi/50$) & 0.999970735& 190(0.005532)& >0.99999999& 34(0.000765)& >0.99999999& 833(0.052076)\\ 
      FA ($\pi/100$) & 0.999970816& 380(0.006922)& >0.99999999& 42(0.001092)& >0.99999999& 858(0.076022)\\ 
      FA ($\pi/150$) & 0.999970877& 570(0.007439)& >0.99999999&	44(0.001514)& >0.99999999& 884(0.051781)\\ 
      \bottomrule
      \multicolumn{5}{l}{ Error tolerance = $1e^{-2}$} & \\
      \bottomrule
      
      Simulink(auto)& 0.110072146& 53(0.300692)& 0.482116372& 51 (0.176899)& 0.914488559	& 55 (0.750841)\\
      Simulink(0.1)& 0.454437239& 505(0.339678)& 0.480452864& 71 (0.191814)& 0.97553497& 205(0.785034)\\
      Simulink(0.01)& 0.979093705& 5008(0.435635)& \textbf{>0.99999999}& \textbf{703 (0.181399)}& 0.99746838& 2005(0.805409)\\
      Simulink(0.001)&\textbf{0.999792386}& \textbf{50007(1.094001)}& >0.99999999& 7003(0.254285)& \textbf{0.999745788}& \textbf{20004(0.989653)}\\
      
      FA ($\pi/10$) & \textbf{0.999969946}& \textbf{38(0.000969)}& \textbf{>0.99999999}& \textbf{16(0.000473)}& \textbf{0.999999732}& \textbf{115(0.008349)}\\ 
      FA ($\pi/50$) & 0.999970735& 190 (0.005691)& >0.99999999 & 24(0.000782)& 0.999999563& 136(0.006016)\\ 
      FA ($\pi/100$) & 0.999970816& 380 (0.007241)& >0.99999999& 34(0.001135)& 0.999999478& 166(0.008188)\\ 
      FA ($\pi/150$) & 0.999970877& 570 (0.002849)& >0.99999999& 44(0.001381)& 0.999999468& 198(0.00909)\\ 
    \end{tabular}
  }
  \caption{\scriptsize Benchmark results. Bold values indicate the fewest number of
    simulation steps needed for each simulation technique resulting in
    accuracy (correlation coefficient) > 0.999 }
  \label{tab:results}
\end{table}
Table~\ref{tab:results} summarizes the results of the benchmarks. We
compare simulation results of \\ Simulink/Stateflow\textregistered\ and \ac{FA} for each
benchmark with different parameters. For Simulink/Stateflow\textregistered\ models, we
simulated the benchmarks with different maximum time step sizes. For the
\ac{FA} models, we simulated the benchmarks with different maximum angle
increments. We also simulated all models with three different error
tolerance settings as shown in Table~\ref{tab:results}. The main
conclusion that can be drawn from Table~\ref{tab:results} is that in
general the proposed simulation technique via conversion of \ac{HA} to
\ac{FA} takes fewer number of simulation steps (is faster) for the same
or higher accuracy. This result is further expounded in
Figure~\ref{fig:step_number}.

\begin{figure}[!htbp]
  \centering
  \hspace{-0.1\textwidth}
  \begin{subfigure}[b]{0.25\textwidth}
    \begin{tikzpicture}
      \begin{axis}[
        width=5cm, height=4cm,
        xmin=30, xmax=50010,
        ymin=0, ymax=1,
        scaled x ticks=false,
        xlabel={\scriptsize Number of simulation steps}, 
        ylabel={\scriptsize Accuracy},
        ylabel style={at={(axis description cs:0.2,.5)}, anchor=south},
        label style={font=\bfseries\boldmath},
        tick label style={font=\bfseries\boldmath, font=\tiny},
        legend pos= south east,
        legend style={font=\tiny},
        ]
        \addplot[  mark = *, blue] 
        table [x=step_sl, y=eff_sl, col sep=comma] {./data/compare.csv};
        \addlegendentry{Simulink/Stateflow\textregistered\ }
        \addplot[ mark = *, red] 
        table [x=step_fa, y=eff_fa, col sep=comma] {./data/compare.csv};
        \addlegendentry{FA}
      \end{axis}
    \end{tikzpicture}
    \caption{\vspace{1mm} \scriptsize Steering wheel system}
    \label{fig:sl_step_first}
  \end{subfigure}%
  \hspace{0.08\textwidth}
  \begin{subfigure}[b]{0.25\textwidth}
    \begin{tikzpicture}
      \begin{axis}[
        width=5cm, height=4cm,
        xmin=30, xmax=1005,
        ymin=0.4, ymax=1,
        scaled x ticks=false,
        xlabel={\scriptsize Number of simulation steps}, 
        ylabel={\scriptsize Accuracy},
        ylabel style={at={(axis description cs:0.2,.5)}, anchor=south},
        label style={font=\bfseries\boldmath},
        tick label style={font=\bfseries\boldmath, font=\tiny},
        legend pos=south east,
        legend style={font=\tiny},
        ]
        \addplot[  mark = *, blue] 
        table [x=step_sl_two, y=eff_sl_two, col sep=comma] {./data/compare.csv};
        \addlegendentry{Simulink/Stateflow\textregistered\ }
        \addplot[ mark = *, red] 
        table [x=step_fa_two, y=eff_fa_two, col sep=comma] {./data/compare.csv};
        \addlegendentry{FA}
      \end{axis}
    \end{tikzpicture}
    \caption{\vspace{1mm} \scriptsize Robot Control}
    \label{fig:sl_step_second}
  \end{subfigure}%
  \hspace{0.08\textwidth}
  \begin{subfigure}[b]{0.25\textwidth}
    \begin{tikzpicture}
      \begin{axis}[
        width=5cm, height=4cm,
        xmin=30, xmax=2007,
        ymin=0.8, ymax=1,
        scaled x ticks=false,
        xlabel={\scriptsize Number of simulation steps}, 
        ylabel={\scriptsize Accuracy},
        ylabel style={at={(axis description cs:0.2,.5)}, anchor=south},
        label style={font=\bfseries\boldmath},
        tick label style={font=\bfseries\boldmath, font=\tiny},
        legend pos= south east,
        legend style={font=\tiny},
        ]
        \addplot[  mark = *, blue] 
        table [x=step_sl_three, y=eff_sl_three, col sep=comma] {./data/compare.csv};
        \addlegendentry{Simulink/Stateflow\textregistered\ }
        \addplot[ mark = *, red] 
        table [x=step_fa_three, y=eff_fa_three, col sep=comma] {./data/compare.csv};
        \addlegendentry{FA}
      \end{axis}
    \end{tikzpicture}
    \caption{\vspace{1mm} \scriptsize Water heating system}
    \label{fig:sl_step_third}
  \end{subfigure}
  
  \caption{\scriptsize Accuracy vs. Number of Steps between \ac{FA} and Simulink/Stateflow\textregistered\ }
  \label{fig:step_number}
\end{figure}
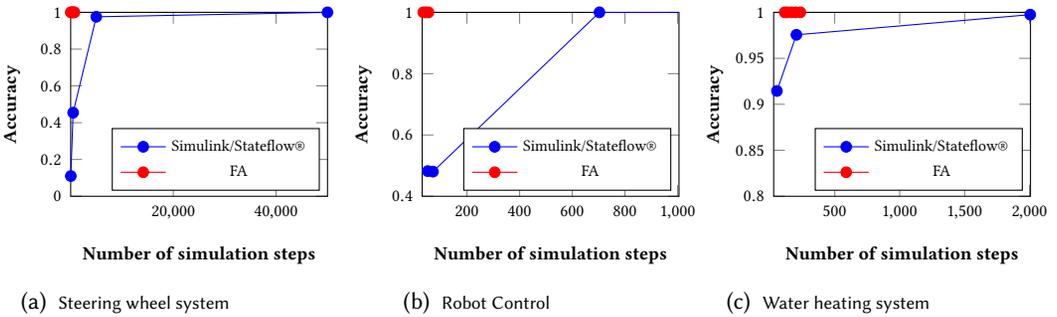

Figure~\ref{fig:step_number} shows the comparison between the accuracy
and the number of simulation steps for \ac{FA} and Simulink/Stateflow\textregistered\
(error bound of $1e^{-2}$ extracted from Table~\ref{tab:results}). The results of \ac{FA} are shown with red circles, and the results of
Simulink/Stateflow\textregistered\ are shown with blue circles.
Figures~\ref{fig:sl_step_first},~\ref{fig:sl_step_second}
and~\ref{fig:sl_step_third} are the results from the first, the second
and the third benchmark, respectively. The results show that, for all
benchmarks, \ac{FA} can generate the outputs with an accuracy close to
1, while for the same number of the simulation steps,
Simulink/Stateflow\textregistered\ generates the output with significantly lower
accuracy.

\subsection{Discussion on overall results --- \ac{FA} vs Simulink/Stateflow\textregistered}
\label{sec:discussion-1}

Overall, the best results, in Table~\ref{tab:results} (in bold), show
that \ac{FA} successfully detected all level crossings with high
accuracy with the significantly fewer steps. For the first benchmark (Figure~\ref{fig:sl_step_first}), Simulink/Stateflow\textregistered\ generates the output with a similar accuracy
(correlation coefficient of > 0.999) as the \ac{FA} with 50007 steps and
1.094001 $\sec$ execution time. On the other hand, the simulation of
\ac{FA}, resulting in correlation coefficient of > 0.9999, is
$1129\times$ faster with a runtime of 0.000969 $\sec$ taking only 38 steps.
For the results of the second benchmark, \ac{FA} generates the output $405\times$ faster with a runtime of 0.000473 $\sec$ taking 16 steps for the same accuracy (correlation coefficient > 0.99999999). Similarly, for the third benchmark, \ac{FA} generates the output (correlation coefficient of
0.999999732) $118\times$ faster with a runtime of 0.008349 $\sec$, taking only 115 steps .
This shows that \ac{FA} is much more efficient in the
number of simulation steps and the execution time. Furthermore, the
guard detection is accurate in \ac{FA}, as the angular rotation always
converges to the guard; this allows precise guard detection.
Furthermore, \ac{FA} can detect the equality guard condition, as shown
with the water heating benchmark, which Simulink/Stateflow\textregistered\ could not
detect even with a very fine resolution for the time step.


\section{Related Work --- comparison with state-of-the-art \ac{HA}
  simulation techniques}
\label{sec:related-work}

Multiple simulation techniques have been proposed for correctly handling
even number of level-crossing detection
problem~\cite{park1996state,esposito2001accurate,malik2020dynamic,ro2019compositional}.
The key in all these works is to make the integration/simulation step
size selection sensitive to the level-crossing guard conditions.

The techniques in~\cite{park1996state,esposito2001accurate} maybe too
inefficient in practice as stated in~\cite{zhang2008zero}, because they
require performing extra computations during classical numerical
integration by incorporating Lie derivatives of the guard set.
Furthermore, the most flexible technique~\cite{esposito2001accurate}
cannot handle cases when the \ac{ODE} flow is tangential to the guard
set. Similarly the simulation techniques
in~\cite{malik2020dynamic,ro2019compositional} require performing
optimisation for computing the Lagrangian error bounds at \textit{every}
integration step. To conclude, classical level-crossing detection
techniques are unable to correctly handle discontinuities during
simulation of hybrid systems. A relatively recent \ac{QSS} integration
technique~\cite{fernandez2014stand} does select quantum depending upon
the level-crossing conditions. However, experimenting with the tool
in~\cite{fernandez2014stand} shows that this tool cannot handle:
\textcircled{1} equality in guard conditions, \textcircled{2}
conjunction and disjunction on guard conditions, and \textcircled{3}
non-affine guard conditions, like we do. Furthermore, this tool also
does not take advantage of constant slope \acp{ODE} to efficiently
simulate systems like our technique does.

There also exist plethora of techniques to over-approximate the
trajectory of the hybrid system and the guard set used in verification
of safety properties of hybrid systems~\cite{chen2013flow,
  althoff2014reachability, benvenuti2014assume, SpaceEx}. In fact, the
work in~\cite{chen2013flow} uses Taylor polynomials to over-approximate
the hybrid system trajectories. The proposed technique is orthogonal to
the techniques developed in the verification community, because we
concentrate on simulating complex hybrid systems rather than
verification.


\section{Conclusions and future work}
\label{sec:concl-future-work}
Formal models, such as \acf{STL} and \acf{HA} are central to specification, 
verification and simulation of \acf{CPS}. This is since, \ac{CPS} exhibit 
hybrid dynamics and need methods for effective control. While engineers have used 
both time domain and frequency domain models to tackle the challenges of \ac{CPS}, especially 
 for control and signal processing applications, computer scientists have 
 primarily focussed on time domain models. 
 While, \acf{TFL}~\cite{donze2012temporal} is a significant step in the direction of combined time, 
 frequency-based approach, it is mainly developed as a specification language for 
  properties rather than for effective system modelling and simulation. 

  We bridge this gap, by developing the first formal model for system design, called 
  \acfp{FA}, that combines time and frequency domain modelling, especially suitable for hybrid systems.
  We provide a sound structural translation from \ac{HA} to \ac{FA}. Moreover, we show that 
  \ac{FA} can solve a number of thorny issues related to the effective simulation of 
  \ac{HA}, such as the problem of \emph{precise level crossing detection} and efficient simulation. 
  We show, for the first time, that we can precisely detect every level crossing of an \ac{HA}.
  Moreover, we empirically show that the total number of steps needed for simulation is 
  significantly improved relative to other state of the art simulation schemes in the time domain.
  
  While our work paves the way for an unified approach of time and
  frequency domain modelling and simulation, we have some limitations of
  the framework. First, we have restricted the \ac{HA} to capture
  dynamics of the form $\dot{x}(t) = f(x(t))$. We need to generalise
  this to capture other types of complex dynamics. Second, this paper
  primarily focuses on simulation. However, given the efficiency of
  simulation in the frequency domain, it may be of interest to verify
  \ac{TFL} properties over \ac{FA} specifications. We will shortly
  explore these problems as future work.


\bibliographystyle{ACM-Reference-Format}
\bibliography{refb}


\end{document}